\documentclass[11pt,letterpaper]{article}
\usepackage[utf8]{inputenc}
\usepackage[T1]{fontenc}
\usepackage{microtype}
\usepackage{lmodern}

\usepackage[margin=1in]{geometry}

\usepackage{booktabs}
\usepackage{amsmath}
\usepackage{amsfonts}
\usepackage{amssymb}
\usepackage{amsthm}
\usepackage{amstext}

\usepackage{xcolor}
\usepackage{graphicx}

\usepackage{mathtools}
\usepackage{xspace}
\usepackage{enumitem}
\usepackage[ruled,vlined]{algorithm2e}

\usepackage{bm}

\usepackage{url}
\usepackage[colorlinks,breaklinks]{hyperref} %
\hypersetup{citecolor=green!60!black,linkcolor=blue!80!black,urlcolor=blue!80!black}
\hypersetup{pdftitle={Optimal Online Weighted Connectivity Augmentation},pdfauthor={Mohit Garg, Felix Hommelsheim, Alexander Lindermayr}}

\usepackage[skip=5pt plus 3pt minus 1pt,parfill=0pt]{parskip}
\makeatletter
\def\thm@space@setup{%
  \thm@preskip=\parskip \thm@postskip=0pt
}
\renewcommand\paragraph{\@startsection{paragraph}{4}{\z@}{\parskip}{-1em}{\normalfont\normalsize\bfseries}}
\makeatother

\allowdisplaybreaks

\usepackage{aliascnt}
\usepackage[capitalize,nameinlink,noabbrev]{cleveref}
\Crefname{algocf}{Algorithm}{Algorithms}

\usepackage[
 backend=biber,
 style=alphabetic,
 citetracker,
 hyperref=auto,
 maxcitenames=5,
 sortcites,
 sorting=anyt,
 maxbibnames=12,
 date=year,
 isbn=false,
 url=false,
 doi=false,
 eprint=false,
]{biblatex}

\newbibmacro{string+doiurlisbn}[1]{%
  \iffieldundef{doi}{%
    \iffieldundef{url}{%
      #1%
    }{%
      \href{\thefield{url}}{#1}%
    }%
  }{%
    \href{https://doi.org/\thefield{doi}}{#1}%
  }%
}

\DeclareFieldFormat
[article,inbook,incollection,inproceedings,patent,thesis,unpublished]
  {title}{\usebibmacro{string+doiurlisbn}{#1}}

\definecolor{col1}{HTML}{56b3e9}
\definecolor{col2}{HTML}{e69f00}
\definecolor{col3}{HTML}{009e74}
\definecolor{col4}{HTML}{cc79a7}
\definecolor{col5}{HTML}{d55e00}
\definecolor{col6}{HTML}{0071b2}

\usepackage{tikz}
\usetikzlibrary{arrows.meta}
\usetikzlibrary{calc,backgrounds}

\colorlet{paperBought}{col6}
\colorlet{paperInterval}{col2}

\colorlet{paperSecondInterval}{col4}
\colorlet{paperNeutral}{black!55}

\tikzset{
    vertex/.style={circle, draw, fill=black, inner sep=2.5pt},
    comp/.style={circle, draw, line width=1.3pt, inner sep=3.5pt},
    edge/.style={line width=1.1pt},
    paper/diagram/.style={font=\small, line cap=round, line join=round},
    paper/vertex/.style={
        circle, draw=black, fill=black, inner sep=0pt,
        minimum size=3.5pt, outer sep=0pt
    },
    paper/terminal/.style={
        paper/vertex, fill=white, line width=0.8pt, minimum size=5.5pt
    },
    paper/edge/.style={draw=black!55, line width=0.7pt},
    paper/link/.style={
        draw=black!50, line width=0.7pt, dash pattern=on 3pt off 2pt
    },
    paper/bought link/.style={
        paper/link, draw=paperBought, solid, line width=1.2pt
    },
    paper/interval/.style={
        draw=paperInterval, fill=paperInterval!12,
        line width=0.7pt, dash pattern=on 2pt off 1.5pt
    },
    paper/second interval/.style={
        paper/interval, draw=paperSecondInterval, fill=paperSecondInterval!12,
        dash pattern=on 0.3pt off 1.4pt, line cap=round
    },
    paper/past interval/.style={fill=none},
    paper/interval bar/.style={line width=1.1pt},
    paper/maximal interval bar/.style={paper/interval bar, line width=1.7pt},
    paper/pivot guide/.style={line width=0.5pt, densely dotted},
    paper/boundary guide/.style={
        draw=black!30, line width=0.4pt, dash pattern=on 2pt off 2pt
    },
    paper/optimal link/.style={draw=black!80, line width=1pt},
    paper/internal optimal link/.style={paper/optimal link, draw=col4},
    paper/rooted optimal link/.style={paper/optimal link},
    paper/secondary optimal link/.style={dash pattern=on 3pt off 2pt},
    paper/tertiary optimal link/.style={dash pattern=on 0pt off 3pt, line cap=round},
    paper/label/.style={inner sep=1pt, outer sep=0pt},
    paper/edge label/.style={paper/label, font=\footnotesize},
    paper/annotation/.style={paper/label, font=\footnotesize},
    paper/transition/.style={
        draw=black!65, line width=0.7pt, -{Stealth[length=4pt, width=3.5pt]}
    }
}

\newtheorem{theorem}{Theorem}
\newaliascnt{corollary}{theorem}

\aliascntresetthe{corollary}
\newaliascnt{lemma}{theorem}
\newtheorem{lemma}[lemma]{Lemma}
\aliascntresetthe{lemma}
\newaliascnt{proposition}{theorem}
\newtheorem{proposition}[proposition]{Proposition}
\aliascntresetthe{proposition}
\theoremstyle{definition}
\newaliascnt{definition}{theorem}
\newtheorem{definition}[definition]{Definition}
\aliascntresetthe{definition}

\newcommand{\ex}{\mathbb E}

\newcommand{\cC}{\mathcal{C}}
\newcommand{\cA}{\mathcal{A}}
\newcommand{\cI}{\mathcal{I}}
\newcommand{\cT}{\ensuremath{\mathcal{T}}\xspace}

\newcommand{\cost}{\mathsf{cost}\xspace}
\newcommand{\eps}{\varepsilon}

\DeclareMathOperator*{\argmin}{arg\,min}

\usepackage[colorinlistoftodos,prependcaption,textsize=scriptsize]{todonotes}

\title{Optimal Online Weighted Connectivity Augmentation}
\date{}

\author{Mohit Garg\thanks{The Institute of Mathematical Sciences, Chennai, India; Homi Bhabha National Institute, Mumbai, India (\texttt{garg@imsc.res.in}).}\and Felix Hommelsheim\thanks{Institute for Computer Science, University of Cologne, Germany (\texttt{hommelsheim@cs.uni-koeln.de}). Funded by the Deutsche Forschungsgemeinschaft (DFG, German Research Foundation) -- Project Number 573939419.}\and Alexander Lindermayr\thanks{Institut für Mathematik, Technische Universität Berlin, Germany (\texttt{alexander.lindermayr@tu-berlin.de}). Supported by the Deutsche Forschungsgemeinschaft (DFG, German Research Foundation) under Germany’s Excellence Strategy -- The Berlin Mathematics Research Center MATH$^+$ (EXC-2046/1, EXC-2046/2, project ID: 390685689).}}

\begin{document}

\pagenumbering{roman}
\maketitle

\begin{abstract}
In the online Weighted Connectivity Augmentation Problem (WCAP), we are given a $k$-edge-connected graph on $n$ vertices and a set of additional edges, called links, with nonnegative costs. Pairs of vertices arrive one by one, and an online algorithm must immediately and irrevocably buy links so that each pair becomes $(k+1)$-edge-connected. The goal is to minimize the total cost of the purchased links.

We obtain a deterministic polynomial-time $O(\log n)$-competitive algorithm for WCAP, improving the previous $O(\log^2 n)$ bound (SODA 2026). Our bound is asymptotically optimal, even for randomized algorithms in the random-order model. The general problem reduces to the case where the underlying graph is a cycle, for which link costs in $\{0,1\}$ already yield an $\Omega(\log n)$ lower bound. Surprisingly, for unit costs on a cycle, we obtain a deterministic polynomial-time $4$-competitive algorithm, improving the previous $O(\log n)$ bound (SODA 2026), and prove that this ratio is optimal among deterministic algorithms. Thus, with unit costs, cycles are easier than trees, where the optimal competitive ratio is $\Theta(\log n)$.

Unlike on a tree, the minimum cuts of a cycle avoiding a fixed vertex need not form a laminar family. Our algorithm for general costs constructs a laminar family of minimum cuts online, decomposes it on the fly into chains, and covers each chain using a primal-dual online algorithm. Combining these algorithms directly would yield only an $O(\log^2 n)$ competitive ratio. 
Our main technical ingredient transforms any deterministic competitive algorithm for covering a chain into one with a stronger, refined cost guarantee.
Together with a structural property of the chain decomposition, this transformation allows us to achieve an $O(\log n)$ competitive ratio.

For unit costs, our algorithm greedily grows sets of purchased links connected through crossings or shared endpoints. The geometry of the cycle allows us to charge each expansion of such a set toward a requested vertex to a distinct link of an optimal solution.
\end{abstract}

\thispagestyle{empty}

\newpage
\pagenumbering{arabic}

\section{Introduction}

Survivable network design seeks to construct low-cost networks that maintain a desired level of connectivity despite the failure of a few edges or vertices.
A fundamental problem in this area is the \emph{Weighted Connectivity Augmentation Problem} (WCAP):
Given a $k$-edge-connected graph $G=(V,E)$ on $n$ vertices and a set $L$ of additional edges on~$V$, called \emph{links}, with nonnegative costs, the goal is to find a minimum-cost subset of links whose addition to~$G$ makes it $(k+1)$-edge-connected.
WCAP and its unit-cost version, the \emph{Connectivity Augmentation Problem} (CAP), have been extensively studied in approximation algorithms (see \Cref{sec:related}).
Connectivity augmentation is also a key ingredient in approximation algorithms for more general survivable network design problems~\cite{WilliamsonGMV95}.

In online WCAP, the $k$-edge-connected graph $G$ and the set of links $L$ with nonnegative costs are known upfront, but requests for additional connectivity between pairs of vertices arrive one by one.
Upon receiving a request $\{s,t\}$, an algorithm must immediately and irrevocably buy zero or more links so that $s$ and~$t$ become $(k+1)$-edge-connected.
Equivalently, every $s$--$t$ cut of~$G$ of size $k$ must be \emph{covered} by a bought link, that is, a link with exactly one endpoint on each side of the cut. The competitive ratio of an algorithm is the worst-case ratio between the (expected) cost of the links purchased by the algorithm and the minimum cost of a set of links that serves all requests.

An online algorithm of Gupta, Krishnaswamy, and Ravi~\cite{GuptaKR12} for a more general survivable network design problem yields a randomized $\widetilde{O}(k\log^3 n)$-competitive algorithm for online WCAP.
For $k=1$, the graph~$G$ can be assumed to be a tree, since contracting its $2$-edge-connected components preserves its minimum cuts.
This special case is known as the \emph{Weighted Tree Augmentation Problem} (WTAP), or simply TAP when all link costs are one.
For online WTAP, Naor, Umboh, and Williamson~\cite{NaorUW22} gave a deterministic $O(\log n)$-competitive algorithm, matching an $\Omega(\log n)$ lower bound for randomized algorithms~\cite{GuptaKR12}, which holds even in the random-order model~\cite{GargS26}.
This remains the best known lower bound for arbitrary~$k$, even for CAP.
Recently, Garg and Subramanian~\cite{GargS26} gave an optimal deterministic $O(\log n)$-competitive algorithm for online CAP, but their algorithm for online WCAP achieves only an $O(\log^2 n)$ competitive ratio and relies on a deterministic algorithm for online node-weighted Steiner forest~\cite{BorstEV25}.
Since online node-weighted Steiner forest generalizes online set cover, it admits no deterministic $o(\log^2 n/\log\log n)$-competitive algorithm~\cite{AlonAABN09}. Hence, a black-box application of an algorithm for general online node-weighted Steiner forest cannot close this logarithmic gap.
Prior to this work, it remained open whether online WCAP for $k\ge 2$ admits an $O(\log n)$-competitive algorithm, as in the case of online WTAP and online CAP.

Our approach is based on reducing online WCAP to the special case of cycles.
The minimum cuts of a $k$-edge-connected graph can be represented by a \emph{cactus}, a connected graph in which every edge belongs to exactly one cycle~\cite{DinitzKL76}.
In the offline setting, one can repeatedly unfold the cactus at vertices shared by two cycles, merging them into a single cycle in which the shared vertex appears twice, with a zero-cost link joining the two copies.
This reduces WCAP to the special case where $k=2$ and $G$ is a cycle, known as the \emph{Weighted Cycle Augmentation Problem} (WCycAP)~\cite{GalvezGAS21,TraubZ23}.
Garg and Subramanian~\cite{GargS26} showed that the reduction to a cactus extends to the online setting, and we show that the unfolding does as well.
Thus, online WCycAP captures the full difficulty of online WCAP.
See \Cref{fig:reduction} for an illustration.

Cycles already exhibit the key difficulty beyond trees.
On a tree, fix a root and identify each minimum cut with the set of vertices on the side not containing the root.
These sets form a \emph{laminar} family: any two are either disjoint or nested.
The algorithm of Naor, Umboh, and Williamson~\cite{NaorUW22} for online WTAP exploits this laminar structure.
In contrast, on a cycle, the corresponding sets are intervals of a path, which may overlap without being nested.

Cycles also raise the second question of whether the logarithmic competitive ratio arises from the link costs or from the structure of the cuts themselves.
On trees, the $\Omega(\log n)$ lower bound holds even with unit costs~\cite{GuptaKR12,GargS26}.
On cycles, however, the reduction from CAP establishes this lower bound only for link costs in~$\{0,1\}$, since it introduces zero-cost links.
For the unit-cost \emph{Cycle Augmentation Problem} (CycAP), the best known competitive ratio was $O(\log n)$~\cite{GargS26}, leaving open whether a constant competitive ratio is achievable.

\subsection{Our Results}

We resolve both questions.
Our main result closes the logarithmic gap for online WCAP.

\begin{theorem}\label{thm:main-weighted-connectivity}
  There exists a deterministic polynomial-time $O(\log n)$-competitive algorithm for the online Weighted Connectivity Augmentation Problem (WCAP) on graphs with $n$ vertices.
\end{theorem}

This improves the previous $O(\log^2 n)$ bound of Garg and Subramanian~\cite{GargS26} and is asymptotically optimal, even for randomized algorithms in the random-order model.
Indeed, the $\Omega(\log n)$ lower bound for $k=1$~\cite{GuptaKR12,GargS26} extends to every fixed~$k$: replacing each edge of a tree by $k$ parallel edges yields a $k$-edge-connected graph with the same family of minimum cuts.
Thus, online WCAP has an optimal competitive ratio of $\Theta(\log n)$ for every~$k$.

By the online reduction of \Cref{prop:connectivity-to-cycle}, which produces a cycle with $O(n)$ vertices, \Cref{thm:main-weighted-connectivity} follows from our next result.

\begin{theorem}\label{thm:main-weighted}
    There exists a deterministic polynomial-time $O(\log n)$-competitive algorithm for the online Weighted Cycle Augmentation Problem (WCycAP) on cycles with $n$ vertices.
\end{theorem}

Since online WCAP reduces to online WCycAP, the latter also inherits the $\Omega(\log n)$ lower bound for link costs in~$\{0,1\}$.
For unit costs, however, we can surprisingly eliminate the logarithmic factor entirely.

\begin{theorem}\label{thm:main-unweighted}
   There exists a deterministic polynomial-time $4$-competitive algorithm for the online Cycle Augmentation Problem (CycAP).
\end{theorem}

This improves the previous $O(\log n)$ bound of Garg and Subramanian~\cite{GargS26}.
Thus, with unit costs, online CAP is easier on cycles than on trees, whereas allowing link costs in~$\{0,1\}$ already makes cycles as hard as general online WCAP.
We complement \Cref{thm:main-unweighted} with a matching lower bound.

\begin{theorem}\label{thm:main-unweighted-lower-bound}
   No deterministic online algorithm for CycAP has a competitive ratio smaller than~$4$, even if all requests share a common vertex.
\end{theorem}

\Cref{tab:results} summarizes the previous and new bounds.

\begin{table}[h]
    \centering    \caption{Competitive ratios of online augmentation problems. Upper bounds are deterministic, and lower bounds hold for randomized algorithms, even in random order~($\dagger$), or for deterministic algorithms~($\ddagger$). The lower bound for WCycAP follows from \cite{GuptaKR12,GargS26} by \Cref{prop:connectivity-to-cycle}.}
    \begin{tabular}{@{}ll@{\hspace{0.4em}}ll@{\hspace{0.4em}}lll@{}}
        \toprule
        & \multicolumn{4}{c}{Previous work} & \multicolumn{2}{c}{\textbf{This paper}} \\
        \cmidrule(lr){2-5} \cmidrule(l){6-7}
        Problem & \multicolumn{2}{l}{UB} & \multicolumn{2}{l}{LB} & UB & LB \\
        \midrule
        TAP, WTAP & $O(\log n)$ & \cite{NaorUW22} & $\Omega(\log n)^\dagger$ & \cite{GuptaKR12,GargS26} & & \\
        WCAP & $O(\log^2 n)$ & \cite{GargS26} & $\Omega(\log n)^\dagger$ & \cite{GuptaKR12,GargS26} & $\bm{O(\log n)}$ & \\
        CAP & $O(\log n)$ & \cite{GargS26} & $\Omega(\log n)^\dagger$ & \cite{GuptaKR12,GargS26} & & \\
        WCycAP & $O(\log^2 n)$ & \cite{GargS26} & & & $\bm{O(\log n)}$ & $\Omega(\log n)^\dagger$ \\
        CycAP & $O(\log n)$ & \cite{GargS26} & & & $\bm{4}$ & $\bm{4}^\ddagger$ \\
        \bottomrule
    \end{tabular}
    
    \label{tab:results}
\end{table}

\subsection{Technical Overview}\label{sec:overview}

We begin with an overview of the approach behind our main result.

\subsubsection{Weighted Cycle Augmentation}

We are given an underlying cycle~$C$ and a set of weighted links. Pairs of vertices arrive as requests, and upon receiving each request, the algorithm must buy links so that the requested pair becomes $3$-edge-connected. Equivalently, for each request $\{s,t\}$, the algorithm must cover every minimum cut (of size~$2$) separating $s$ and~$t$, where a link covers a cut if its endpoints lie on opposite sides. We fix a vertex~$r$ of the cycle and represent each cut by the side that does not contain~$r$. The minimum cuts separating $s$ and~$t$ are then represented by the nonempty intervals of the path~$C-r$ containing exactly one of $s$ and~$t$.

At a high level, our approach for WCycAP is closely related to that of Naor, Umboh, and Williamson~\cite{NaorUW22} for WTAP, but with some crucial differences. They decompose the underlying tree into paths, which yields instances of the Weighted Path Augmentation Problem (WPAP), the special case of WTAP in which the tree is a path, and run a primal-dual algorithm on each of them. The minimum cuts of a tree, represented by the sides not containing a fixed root, form a laminar family, which is known upfront. This allows them to fix the path decomposition before any requests arrive. In contrast, the minimum cuts of a cycle, represented by intervals of~$C-r$, may overlap without being nested. Nevertheless, our algorithm for WCycAP constructs a laminar family of intervals of~$C-r$ online and buys links to cover them. Since this family depends on both the requests and the links purchased by the algorithm, it cannot be determined in advance. Handling this uncertainty is one of the main challenges compared to WTAP.

\paragraph{Uncrossing online.}

Let us briefly explain how our algorithm constructs the laminar family online (see \Cref{fig:wcycap-request} for an illustration). Upon receiving a request $\{s,t\}$, our algorithm repeatedly covers an inclusion-wise minimal uncovered interval separating $s$ and~$t$, which we call \emph{exposed}. We show that the exposed intervals form a laminar family, even across different requests (\Cref{lem:exposed-laminar}). The intuition is that an interval exposed later was already uncovered when an earlier interval was exposed.

\paragraph{Decomposition into chains.}
Similarly to the decomposition of a tree into paths for WTAP, we split the laminar family of exposed intervals online into chains, that is, families of pairwise nested intervals. Again, since the exposed intervals are revealed only online, we cannot decompose this family upfront. Instead, we fix upfront a balanced binary search tree~$T$ on the vertices of~$C-r$, whose in-order traversal follows their order along the path~$C-r$ (see \Cref{fig:pivot-chains}). We assign every incoming exposed interval to its unique vertex of minimum depth in~$T$, called its \emph{pivot}. Exposed intervals with the same pivot~$p$ contain~$p$, so by laminarity they form a chain~$\cC_p$. The search tree~$T$ provides a useful property for charging against an optimal solution. The pivot of an interval is an ancestor in~$T$ of every other vertex in that interval. Moreover, any interval covered by a link~$\ell$ in the optimal solution contains exactly one endpoint of~$\ell$. Thus, its pivot must be an ancestor of an endpoint of~$\ell$ (possibly the endpoint itself). Since~$T$ is balanced, the intervals covered by an optimal link~$\ell$ belong to at most $2\lceil\log_2 n\rceil$ distinct chains (\Cref{lem:active-rooted-count}).

\paragraph{Covering a chain online.}
We cover each chain~$\cC_p$ using a separate instance of an online algorithm for the \emph{rooted chain covering} problem, where requests are promised to be intervals that contain a common vertex~$p$ and form a chain, and each interval must be covered by a purchased link. Since the requested intervals form a chain, we can order them by inclusion and identify them with consecutive edges of a path. Each link then covers the edges of a subpath, just as in WPAP. If the common vertex~$p$ is an endpoint of~$C-r$, all possible intervals containing~$p$ are nested, regardless of which intervals arrive online. Thus, their inclusion order, and hence the corresponding path, is known upfront, as in the approach of~\cite{NaorUW22}. In general, however, the intervals containing~$p$ need not be nested. Although the intervals assigned to~$\cC_p$ are guaranteed to form a chain, this chain is revealed only online.

\paragraph{Representatives instead of pruning.}

Losing at most a factor of~$2$ in the competitive ratio, we assume that all link costs are either zero or powers of~$2$. We buy all zero-cost links upfront and group the links of equal positive cost into \emph{classes}.

Like the basic algorithm of Naor, Umboh, and Williamson~\cite{NaorUW22} for WPAP, our algorithm for rooted chain covering is primal-dual. Upon receiving an uncovered request~$I$, it raises the dual variable of~$I$ until a link covering~$I$, called the \emph{trigger}, becomes tight. Naor, Umboh, and Williamson~\cite{NaorUW22} then buy the trigger. This suffices because they can prune the link set upfront, using the known path, so that each path edge is covered by at most two links per class. This approach is infeasible in our setting, since the path of requests is not known upfront.

Instead, we buy at most four \emph{representatives}: links of the trigger's class that cover~$I$. For each side of~$p$, we select a link whose endpoint inside~$I$ is closest to~$p$ on that side, and a link whose endpoint outside~$I$ is farthest from~$p$ on that side. These representatives have the following property: if a link of the same class covers both~$I$ and a later request, then the representatives also cover that later request, since all later requests contain~$p$ and are nested with~$I$ (\Cref{lem:nested-covers}). Consequently, no link of a given class can cover two requests whose triggers belong to that class. Since the requests form a chain, it follows that every request is covered by at most two triggers per class. This yields an $O(\log n)$-competitive algorithm for rooted chain covering (\Cref{thm:interval-chain-augmentation}).

\paragraph{Avoiding a second logarithmic factor.}
Running this algorithm independently on each chain yields only an $O(\log^2 n)$ competitive ratio: a link in an optimal solution~$F^*$ may cover requests in $O(\log n)$ chains, incurring an $O(\log n)$ factor in each.

However, if an optimal link covers a request in a chain~$\cC_p$ but not its largest requested interval, both endpoints of the optimal link must lie inside that largest interval. By laminarity and the pivot structure, this can happen for at most one chain per optimal link (\Cref{lem:inner-disjoint}).

For each chain~$\cC_p$, partition the links of an optimal solution~$F^*$ covering its requests into two sets: $R^*_p$, consisting of links that cover its largest interval, and $S^*_p$, consisting of the remaining links. Our goal is to cover~$\cC_p$ at cost $O(c(R^*_p)+\log n\cdot c(S^*_p))$. This suffices because each optimal link belongs to $R^*_p$ for at most $O(\log n)$ chains and to $S^*_p$ for at most one chain. Summing over all chains, the total cost is therefore $O(\log n)\cdot c(F^*)$, yielding an $O(\log n)$ competitive ratio.

\paragraph{A black-box conversion.}
Our goal is to obtain the refined cost guarantee mentioned above. Naor, Umboh, and Williamson~\cite{NaorUW22} achieve such a guarantee for WPAP, where links incident to the known root of the path play the role of~$R^*_p$, and this classification is known upfront. For our chains~$\cC_p$, however, the analogous classification depends on the largest requested interval, which is known only at the end of the instance. In particular, there is a link~$\ell^* \in R^*_p$ that covers every request covered by~$R^*_p$, while~$S^*_p$ covers all remaining requests of~$\cC_p$. An algorithm that knew~$\ell^*$ upfront could therefore buy it immediately and pass only the remaining requests to a $\rho(n)$-competitive algorithm for rooted chain covering, paying at most $c(R^*_p)+\rho(n)\cdot c(S^*_p)$ in total. We instead convert any deterministic $\rho(n)$-competitive algorithm for rooted chain covering into one that achieves the desired cost guarantee $O(c(R^*_p)+\rho(n)\cdot c(S^*_p))$, without knowing~$\ell^*$ upfront (\Cref{thm:refined-guarantee}).

The conversion works in phases, each with a fresh copy of the given algorithm and a budget, initially set to zero. In each phase, we pass every request that is not yet covered to the current copy and buy the links it purchases, as long as their total cost stays within the budget. If serving a request would exceed the budget, the phase ends with that request. The conversion then buys a cheapest link covering this request and starts a new phase whose budget is twice the old one, or the cost of this link if that is larger. Moreover, it buys representatives for the largest request~$M$ of the ended phase, chosen among the links that cover~$M$ and cost at most the new budget, so that they cover every later request that one of these links covers. 
We now explain why the total cost of the conversion is $O(c(R^*_p)+\rho(n)\cdot c(S^*_p))$. The key observation is that once the budget is sufficiently large, at most one more phase can end. Intuitively, if a phase ends because of a request covered by~$\ell^*$, the representatives purchased at its end will cover all future requests that~$\ell^*$ could cover. The remaining requests can then be handled by the given algorithm using~$S^*_p$.

More precisely, once the budget reaches $c(R^*_p)+\rho(n)\cdot c(S^*_p)$, a phase can end only if $\ell^*$ covers one of its requests. Indeed, otherwise all requests passed to the current copy are covered by~$S^*_p$, and the given algorithm spends at most $\rho(n)\cdot c(S^*_p)$. In this case, $\ell^*$ also covers the largest request~$M$ of the ended phase, and its cost is at most $c(R^*_p)$, which does not exceed the new budget. Hence, the representatives purchased at the end of this phase cover every later request that $\ell^*$ covers. From then on, only requests covered by~$S^*_p$ are passed to the current copy, so no further phase ends. Thus, the budget increases at most once after reaching $c(R^*_p)+\rho(n)\cdot c(S^*_p)$, and the final budget is $O(c(R^*_p)+\rho(n)\cdot c(S^*_p))$. Since each phase costs at most a constant times its budget and the budgets grow geometrically, the total cost satisfies the same bound.

\subsubsection{Unit-Cost Cycle Augmentation}

For CycAP, our approach is very different. We use the overlap graph~$X$ on the links, in which two links are adjacent if they \emph{cross}, that is, if they share an endpoint or their endpoints alternate along~$C$.
A set of links $F$ serves $\{s,t\}$ if and only if its induced subgraph $X[F]$ of~$X$ contains a path from a link incident to~$s$ to a link incident to~$t$~\cite{GargS26}.
Garg and Subramanian therefore reduce online CycAP to online Steiner forest, which costs a logarithmic factor~\cite{BermanC97} that is unavoidable on general graphs~\cite{ImaseW91}.
Our algorithm instead grows the connected components of~$X[F]$, which we call \emph{components}, greedily by at most two links at a time, and the geometry of the cycle lets us charge these steps to an optimal solution.

\paragraph{Growing components.}
Upon a request $\{s,t\}$, the algorithm first \emph{activates} each of $s$ and~$t$ that is not yet incident to a bought link, that is, it buys a link at this vertex, preferring a link that crosses a bought one.
Then, until the bought links at~$s$ and at~$t$ lie in the same component, it grows the component~$K$ that contains the bought links at~$s$, which pairwise cross since they share~$s$.
If possible, it \emph{merges} $K$ with another component by buying at most two links.
Otherwise, it \emph{expands} $K$ toward~$t$, that is, on each side of~$t$, it buys a link that crosses~$K$ and ends as close to~$t$ as possible, unless a link of~$K$ already ends at least as close to~$t$ on this side (\Cref{fig:expansion}).

\paragraph{Charging to an optimal solution.}
Fix an optimal solution~$F^*$.
Activations buy at most $2|F^*|$ links, since every requested vertex is incident to a link of~$F^*$ and is activated at most once.
Merges and expansions buy at most two links each.
Only activations that start a new component increase the number of components, whereas every merge decreases it, so there are at most as many merges as such activations.
Hence, it suffices to show that each such activation and each expansion adds a link to the set of links of~$F^*$ that are bought or cross a bought link, as this set only grows.
For such an activation, the links of~$F^*$ at the activated vertex are added, as the algorithm prefers links that cross bought ones.
Thus it remains to argue about expansions.

\paragraph{Paying for an expansion.}
A \emph{gap} of a set of links is a maximal circular interval that contains no endpoint of these links.
After an expansion, let $J$ be the gap containing~$t$ of the set consisting of the links of~$K$ and the new links.
As $J$ does not contain~$s$, some link $f\in F^*$ must cover $J$.
Our analysis rests on the geometric fact that if a set of links is connected in~$X$, then every link with exactly one endpoint in one of its gaps crosses one of its links (\Cref{lem:gap}).
As the links of~$K$ together with the new links form such a connected set, $f$ crosses a link of~$K$ or a new link.
By the choice of the new links, no link that crosses~$K$ ends in~$J$, so $f$ crosses a new link~$\ell$.
Moreover, before the expansion, $f$ was not bought and crossed no bought link, as otherwise $\{\ell\}$ or $\{\ell,f\}$ would have merged $K$ with another component (\Cref{lem:expansion}).
This is why merges may consist of two links.
Altogether, the algorithm buys at most $2|F^*|+2|F^*|=4|F^*|$ links.

\paragraph{The matching lower bound.}
For a power of two $m\ge2$, our instance consists of $2m$ \emph{bundles} of six links each, arranged in a cycle such that links of distinct bundles cross if and only if their bundles are consecutive (\Cref{fig:bundle-instance}).
Hence, serving a request requires a link in every bundle on one of the two paths between the bundles of its two vertices.
All requests contain a fixed vertex~$s$ at a bundle~$L_0$, and the first one pairs~$s$ with a vertex at the opposite bundle.
The algorithm thus buys a link in every bundle of one half of this cycle, whereas an offline solution uses one link in every bundle of the other half.
In this other half, the adversary forces three bought links in almost every bundle.
The six links of a bundle join two vertices with three vertices, and the adversary pairs~$s$ first with one of the two and then with one of the three, each time avoiding endpoints of bought links if possible.
Offline, the link between the two requested vertices suffices, whereas the algorithm needs a bought link at each of them, so it ends up with at least three links in the bundle if the bundle already contained one, and at least two otherwise.
Treating the bundles in binary-search order ensures that all but at most $\log_2 m$ of the visited bundles end up with at least three links, since serving a request at a bundle requires bought links in all bundles between it and~$L_0$ on one side (\Cref{fig:bundle-run}).
Thus, the algorithm buys at least $(m+1)+3(m-1)-\log_2 m$ links, whereas $m+1$ links suffice offline, so the ratio tends to~$4$.

\subsection{Related Work}\label{sec:related}

\paragraph{Online connectivity augmentation.}
For WTAP, Gupta, Krishnaswamy, and Ravi~\cite{GuptaKR12} gave a deterministic $O(\log^2 n)$-competitive algorithm and an $\Omega(\log n)$ lower bound for randomized algorithms, and Naor, Umboh, and Williamson~\cite{NaorUW22} closed this gap.
Garg and Subramanian~\cite{GargS26} showed that the $\Omega(\log n)$ lower bound holds even for TAP under uniformly random arrival order.
Naor, Umboh, and Williamson~\cite{NaorUW22} also obtained an $O(\log n)$-competitive algorithm for online WPAP, which is equivalent to online interval set cover and contains the parking permit problem of Meyerson~\cite{Meyerson05}.
More generally, online weighted set cover with $m$ sets and $N$ elements admits a deterministic $O(\log m\log N)$-competitive algorithm, which is nearly best possible for deterministic algorithms~\cite{AlonAABN09}.
Online rooted chain covering is a special case of online set cover with at most $n^2$ sets and elements, so this algorithm is $O(\log^2 n)$-competitive for it.
Our algorithm for rooted chain covering saves one logarithmic factor by exploiting the interval structure.
Buchbinder and Naor~\cite{BuchbinderN09} survey the online primal-dual method for such covering problems.

\paragraph{Online survivable network design.}
Gupta, Krishnaswamy, and Ravi~\cite{GuptaKR12} initiated the study of the online edge-connectivity survivable network design problem.
Their randomized algorithm is $\widetilde O(r_{\max}\log^3 n)$-competitive and runs in time $O(|E|^{r_{\max}})$, where $r_{\max}$ is the maximum connectivity requirement, that is, the largest number of edge-disjoint paths that a request asks for.
Applied to online WCAP, it is $\widetilde O(k\log^3 n)$-competitive.
For $r_{\max}=2$, Naor, Umboh, and Williamson~\cite{NaorUW22} improve the randomized bound to $\widetilde O(\log^2 n)$ using their online WTAP algorithm.
The best known lower bound of $\Omega(\log n)$ already holds for online Steiner tree, that is, for $r_{\max}=1$~\cite{ImaseW91}, so no dependence on~$r_{\max}$ is known to be necessary.
Dehghani, Ehsani, Hajiaghayi, Liaghat, and Seddighin~\cite{DehghaniEHL18} showed that a natural greedy algorithm is $O(\log^2 n\log r_{\max})$-competitive if every connectivity requirement may be violated by a factor of~$2$, and they obtained stronger guarantees for stochastic arrivals.
Qian, Umboh, and Williamson~\cite{QianUW18} developed a framework for online constrained forest problems with proper cut-requirement functions, and its competitive ratio is logarithmic in the number of terminals.
This framework captures online Steiner forest but does not directly apply to connectivity augmentation.

\paragraph{Online Steiner problems.}
Online Steiner tree~\cite{ImaseW91,WestbrookY95} and online edge-weighted Steiner forest~\cite{AwerbuchAB04,BermanC97,WestbrookY95} admit $O(\log n)$-competitive algorithms, and this bound is asymptotically tight~\cite{ImaseW91}.
The node-weighted setting is substantially more difficult.
Hajiaghayi, Liaghat, and Panigrahi~\cite{HajiaghayiLP17} gave the first polynomial-time polylogarithmic-competitive algorithm for online node-weighted Steiner forest.
Borst, Eli\'a\v{s}, and Venzin~\cite{BorstEV25} obtained a randomized $O(\log q\log n)$-competitive algorithm, where $q$ is the number of terminal pairs, and a deterministic $O(\log^2 n)$-competitive algorithm, which underlies the $O(\log^2 n)$ result of \cite{GargS26} for WCAP.

\paragraph{Offline connectivity augmentation.}
A long line of work has led to approximation ratios below~$2$ for CAP and WCAP~\cite{A19,CTZ21,CG18,CG18a,EFKN09,FGKS18,GKZ18,KN16,KN16b,N03,ByrkaGJ23,CecchettoTZ24,TraubZ25,TraubZ23,GalvezGAS21}.
These include a $1.393$-approximation for CAP~\cite{CTZ21} and a $(1.5+\eps)$-approximation for WCAP~\cite{TraubZ23}, which also builds on the reduction to WCycAP.
For CycAP, G\'alvez, Grandoni, Jabal Ameli, and Sornat~\cite{GalvezGAS21} showed APX-hardness and gave a $(1.5+\eps)$-approximation, and for the more general Leaf-to-Leaf Connectivity Augmentation Problem, Cecchetto, Traub, and Zenklusen~\cite{CecchettoTZ24} obtained a $1.29$-approximation.

\section{Preliminaries}\label{sec:preliminaries}

We formalize online WCAP as an online cut-covering problem and reduce it to online WCycAP.

\paragraph{Graphs and cuts.}
All graphs are undirected and may have parallel edges.
For a graph $G=(V,E)$ and a set $F$ of additional edges on $V$, called \emph{links}, we write $G+F:=(V,E\cup F)$.
Edges and links are distinct objects, even when they have the same endpoints.
For a set $Y$ of edges or links on $V$ and a set $S\subseteq V$, let $\delta_Y(S)$ be the set of elements of $Y$ with exactly one endpoint in~$S$, and let $\delta_G(S):=\delta_E(S)$.
The set $S$ \emph{separates} two vertices $u$ and $v$ if $|S\cap\{u,v\}|=1$.
A link \emph{covers} $S$ if exactly one of its endpoints lies in~$S$, and a set of links $F$ \emph{covers} $S$ if $\delta_F(S)\neq\emptyset$.
For an integer $k\ge1$, two distinct vertices $u,v\in V$ are \emph{$k$-edge-connected} in $G$ if $|\delta_G(S)|\ge k$ for every $S\subseteq V$ that separates $u$ and~$v$.
The graph $G$ is \emph{$k$-edge-connected} if every two distinct vertices are $k$-edge-connected in~$G$.

\paragraph{Online weighted connectivity augmentation.}
An instance of online WCAP consists of a $k$-edge-connected graph $G=(V,E)$ with $n$ vertices, a set $L$ of links on $V$ with costs $c\colon L\to\mathbb{R}_{\ge0}$, and a sequence of \emph{requests}, each of which is a set $\{s,t\}$ of two distinct vertices.
A set of links $F\subseteq L$ \emph{serves} a request $\{s,t\}$ if $s$ and $t$ are $(k+1)$-edge-connected in $G+F$.
The graph~$G$ and the links with their costs are known upfront, whereas the requests are revealed one by one.
An online algorithm maintains a set $F\subseteq L$ of bought links.
Upon each request, it must add links to $F$ such that $F$ serves the request, without knowing the future requests or their number, and it can never remove links from~$F$.
The objective is to minimize the total cost $c(F):=\sum_{\ell\in F}c_\ell$.
A set of links is \emph{feasible} for a request sequence if it serves all of its requests.
We consider only request sequences for which $L$ is feasible.
Moreover, we assume that no two links are parallel.
Parallel links cover the same sets, so an algorithm may ignore all but a cheapest one among them without changing the optimal cost.

An online algorithm is \emph{$\rho(n)$-competitive} if, for every instance with $n$ vertices and every feasible solution~$F^*$ of the instance, its solution $F$ satisfies $c(F)\le\rho(n)\cdot c(F^*)$, or $\ex[c(F)]\le\rho(n)\cdot c(F^*)$ if the algorithm is randomized, where an oblivious adversary fixes the instance upfront.
Its \emph{competitive ratio} for instances with $n$ vertices is the smallest such~$\rho(n)$.
In the \emph{random-order model}, the requests of the instance arrive in uniformly random order, and we require $\ex[c(F)]\le\rho(n)\cdot c(F^*)$, where the expectation is also taken over this order.

\paragraph{Requests as cut-covering constraints.}
Since $G$ is $k$-edge-connected, we have $|\delta_G(S)|\ge k$ for every nonempty proper subset $S\subsetneq V$.
We call such a set $S$ a \emph{minimum cut} of $G$ if $|\delta_G(S)|=k$, and then $V\setminus S$ is a minimum cut as well.
The following standard characterization turns every request into a family of covering constraints.

\begin{lemma}\label{lem:cut-covering}
    Let $G=(V,E)$ be a $k$-edge-connected graph, let $F$ be a set of links on~$V$, and let $s,t\in V$ be distinct.
    Then $F$ serves $\{s,t\}$ if and only if $F$ covers every minimum cut of $G$ that separates $s$ and~$t$.
\end{lemma}

\begin{proof}
    By definition, $F$ serves $\{s,t\}$ if and only if $|\delta_{G+F}(S)|\ge k+1$ for every $S\subseteq V$ that separates $s$ and~$t$.
    For every such set~$S$, we have $|\delta_{G+F}(S)|=|\delta_G(S)|+|\delta_F(S)|$ and $|\delta_G(S)|\ge k$.
    Hence, the condition is violated for $S$ if and only if $|\delta_G(S)|=k$ and $\delta_F(S)=\emptyset$.
\end{proof}

\paragraph{Reduction to cycle augmentation.}
Online WCycAP is the special case of online WCAP in which $k=2$ and $G$ is a simple cycle~$C$ on $n\ge3$ vertices.
The minimum cuts of~$C$ are its \emph{circular intervals}, that is, the nonempty proper subsets of $V$ whose vertices are consecutive along~$C$, since exactly two edges of~$C$ leave such a set and at least four edges of~$C$ leave every other nonempty proper subset of~$V$.
The minimum cuts of a $k$-edge-connected graph can be represented by a cactus with $O(n)$ edges~\cite{DinitzKL76}, which can be constructed deterministically in polynomial time~\cite{Fleischer98} and in near-linear time by a randomized algorithm~\cite{KargerP09}.
Garg and Subramanian~\cite{GargS26} use this representation to reduce online WCAP to the online \emph{Weighted Cactus Augmentation Problem} (WCacAP), in which $k=2$ and $G$ is a \emph{cactus}, that is, a connected graph in which every edge lies on exactly one cycle, where cycles of length two are allowed.
In the offline setting, G\'alvez, Grandoni, Jabal Ameli, and Sornat~\cite{GalvezGAS21} reduce WCacAP further to WCycAP by repeatedly unfolding the cactus at a vertex~$x$ that lies on at least two cycles, which joins two of these cycles into one cycle on which $x$ appears as two copies, joined by a zero-cost link (see also~\cite{TraubZ23}).
Combining both reductions gives the following online reduction, which \Cref{fig:reduction} illustrates and which we prove in \Cref{app:reduction}.

\begin{figure}[tbp]
    \centering
\begin{tikzpicture}[
    paper/diagram,
    left cycle/.style={paper/edge, draw=col2, line width=1.2pt},
    right cycle/.style={paper/edge, draw=col3, line width=1.2pt},
    articulation/.style={paper/terminal},
    zero link/.style={paper/link, line width=1.2pt, dash pattern=on 0pt off 2.5pt},
    link label/.style={paper/edge label, text=black!55},
    blob edge/.style={},
    blob node/.style={}
]

    \newcommand{\Kfour}[4][]{%
        \begin{scope}[shift={(#3,#4)}, #1]
            \coordinate (#2n) at ( 0.00, 0.28);
            \coordinate (#2e) at ( 0.28, 0.00);
            \coordinate (#2s) at ( 0.00,-0.28);
            \coordinate (#2w) at (-0.28, 0.00);
            \foreach \u/\v in {
                #2n/#2e,#2n/#2s,#2n/#2w,
                #2e/#2s,#2e/#2w,#2s/#2w%
            } {
                \draw[paper/edge, blob edge] (\u)--(\v);
            }
            \foreach \u in {#2n,#2e,#2s,#2w} {
                \node[paper/vertex, blob node] at (\u) {};
            }
        \end{scope}
    }

    \begin{scope}[scale=1.1]
        \Kfour[blob edge/.style={draw=black}, blob node/.style={paper/terminal, minimum size=4.5pt}]{x}{ 0.00}{ 0.00}
        \node[paper/label, above=5pt] at (xn) {$x$};

        \Kfour{a}{-1.15}{ 0.75}
        \Kfour{b}{-1.45}{-0.32}
        \Kfour{c}{-0.45}{-0.92}

        \Kfour{d}{ 1.15}{ 0.75}
        \Kfour{e}{ 1.45}{-0.32}
        \Kfour{f}{ 0.45}{-0.92}

        \begin{scope}[on background layer]
            \draw[left cycle] (xw) -- (ae);
            \draw[left cycle] (as) -- (bn);
            \draw[left cycle] (be) -- (cw);
            \draw[left cycle] (cn) -- (xs);

            \draw[right cycle] (xe) -- (dw);
            \draw[right cycle] (ds) -- (en);
            \draw[right cycle] (ew) -- (fe);
            \draw[right cycle] (fn) -- (xs);
        \end{scope}

        \draw[paper/link] (an) to[bend left=12]
            node[link label, above=2pt] {$\ell_1$} (dn);
        \draw[paper/link] (ce) to[bend left=30]
            node[link label, pos=0.8, above=2pt] {$\ell_2$} (ew);
        \draw[paper/link] (be) --
            node[link label, pos=0.3, above=2pt] {$\ell_3$} (fn);
    \end{scope}
    \node[paper/label] at (0,-1.95) {(a) WCAP};

    \draw[paper/transition] (2.05,0) --
        node[paper/annotation, above=3pt] {cactus} (2.95,0);

    \begin{scope}[xshift=5.0cm]
        \begin{scope}[scale=1.1]
            \coordinate (x)  at (0,0);
            \coordinate (a)  at (-1.15, 0.75);
            \coordinate (b)  at (-1.45,-0.32);
            \coordinate (c)  at (-0.45,-0.92);
            \coordinate (d)  at ( 1.15, 0.75);
            \coordinate (e)  at ( 1.45,-0.32);
            \coordinate (f)  at ( 0.45,-0.92);

            \draw[left cycle] (x)--(a)--(b)--(c)--cycle;
            \draw[right cycle] (x)--(d)--(e)--(f)--cycle;

            \draw[paper/link] (a) --
                node[link label, above=3pt] {$\ell_1$} (d);
            \draw[paper/link] (c) --
                node[link label, above=3pt] {$\ell_2$} (e);
            \draw[paper/link] (b) --
                node[link label, above=3pt] {$\ell_3$} (f);

            \foreach \u in {a,b,c,d,e,f} {
                \node[paper/vertex] at (\u) {};
            }
            \node[articulation] at (x) {};
            \node[paper/label, above=5pt] at (x) {$x$};
        \end{scope}
        \node[paper/label] at (0,-1.95) {(b) WCacAP};
    \end{scope}

    \draw[paper/transition] (7.05,0) --
        node[paper/annotation, above=3pt] {unfold} (7.95,0);

    \begin{scope}[xshift=10.0cm]
        \begin{scope}[scale=1.1]
            \foreach \i in {1,...,8} {
                \coordinate (v\i) at ({1.50*cos(90-45*(\i-1))},{1.05*sin(90-45*(\i-1))});
            }

            \draw[right cycle] (v1) arc[start angle=90, end angle=-90, x radius=1.50, y radius=1.05];
            \draw[left cycle] (v5) arc[start angle=-90, end angle=-270, x radius=1.50, y radius=1.05];

            \draw[paper/link] (v2) to[bend left=20]
                node[link label, pos=0.18, below=2pt] {$\ell_1$} (v8);
            \draw[paper/link] (v3) to[bend right=5]
                node[link label, pos=0.2, above=2pt] {$\ell_2$} (v6);
            \draw[paper/link] (v4) to[bend right=5]
                node[link label, pos=0.8, above=2pt] {$\ell_3$} (v7);

            \draw[zero link] (v1) -- (v5);

            \foreach \i in {2,3,4,6,7,8} {
                \node[paper/vertex] at (v\i) {};
            }
            \foreach \i in {1,5} {
                \node[articulation] at (v\i) {};
            }
            \node[paper/label, above=5pt] at (v1) {$x_1$};
            \node[paper/label, below=5pt] at (v5) {$x_2$};
        \end{scope}
        \node[paper/label] at (0,-1.95) {(c) WCycAP};
    \end{scope}
\end{tikzpicture}
    \caption{The reduction from WCAP (a) via a cactus (b) to WCycAP (c).
        Exactly two edges, both orange or both green, leave every minimum cut of (a).
        Unfolding duplicates $x$ into $x_1$ and $x_2$, joined by a zero-cost link (dotted).}
    \label{fig:reduction}
\end{figure}
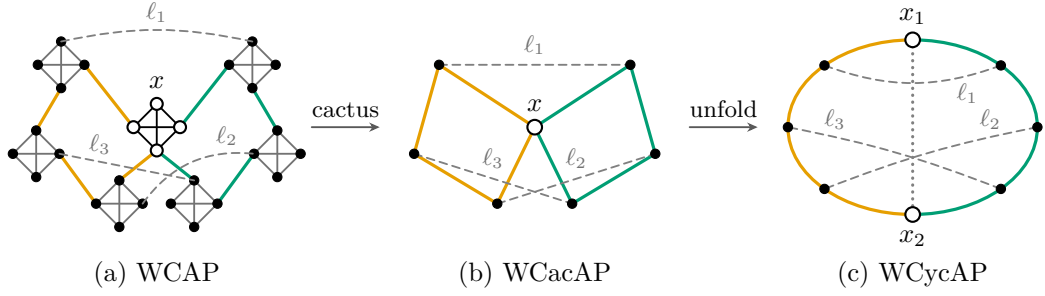

\begin{proposition}[Reduction to Weighted Cycle Augmentation]\label{prop:connectivity-to-cycle}
    There is a constant $a\ge1$ such that, for every nondecreasing function $\rho\colon\mathbb{N}_{\ge1}\to\mathbb{R}_{\ge1}$, a $\rho(n)$-competitive algorithm for online WCycAP yields a $\rho(\lceil an\rceil)$-competitive algorithm for online WCAP.
    The resulting algorithm is deterministic if the given one is, and it runs in polynomial time if the given one does.
\end{proposition}

The reduction maps each request separately, and the resulting algorithm never pays more than the given algorithm for online WCycAP.
Therefore, \Cref{prop:connectivity-to-cycle} also applies to randomized algorithms and to the random-order model, since requests in uniformly random order are mapped to requests in uniformly random order.
As the reduction preserves link costs and adds only links of cost zero, it carries the $\Omega(\log n)$ lower bound for online CAP~\cite{GuptaKR12,GargS26} over to online WCycAP with link costs in~$\{0,1\}$.
Hence, \Cref{thm:main-weighted} is asymptotically optimal, even for randomized algorithms in the random-order model.

\section{Online Weighted Cycle Augmentation}\label{sec:weighted}

In this section, we prove \Cref{thm:main-weighted}.
Throughout this section, $V=\{v_1,\ldots,v_n\}$ is a set of vertices with the linear order $v_1 \prec \cdots \prec v_n$.
A set $I \subseteq V$ is an \emph{interval} if $I = \{v_a,v_{a+1},\ldots,v_b\}$ for some $1 \le a \le b \le n$.
As in \Cref{sec:preliminaries}, a link covers an interval $I$ if exactly one of its endpoints lies in~$I$, $\delta_F(I)$ denotes the set of links in a set $F$ that cover~$I$, and $F$ covers $I$ if $\delta_F(I)\neq\emptyset$.

\paragraph{Cost classes.}
Rounding every positive link cost up to the smallest power of~$2$ that is at least this cost multiplies it by a factor in~$[1,2)$.
Hence, an algorithm that is $\rho(n)$-competitive on instances whose positive costs are powers of~$2$ is $2\rho(n)$-competitive on arbitrary instances if it is run on the rounded costs.
We therefore assume that every link has cost $0$ or $2^j$ for some integer~$j$, and we call $j$ the \emph{class} of the link.
Links of cost zero belong to no class, and our algorithms buy all of them upfront.

\subsection{Splitting Exposed Intervals into Chains}\label{sec:framework}

Let $C = v_1 - v_2 - \cdots - v_n - v_1$ be a simple cycle on~$V$, let $L$ be a set of links on $V$ with nonnegative costs, and let $r:= v_1$.
The cut-covering formulation is standard for cycle augmentation~\cite{TraubZ23}. For a single request, \Cref{lem:cut-covering} gives the following characterization.

\begin{lemma}\label{lem:interval-covering}
    A set of links $F \subseteq L$ serves a request $\{s,t\}$
    if and only if $\delta_F(I) \neq \emptyset$ for every interval $I \subseteq V \setminus \{r\}$ with $|\{s,t\} \cap I| = 1$.
\end{lemma}

\begin{proof}
    The minimum cuts of $C$ are its circular intervals, so by \Cref{lem:cut-covering}, $F$ serves $\{s,t\}$ if and only if $F$ covers every circular interval $S$ with $|S\cap\{s,t\}|=1$.
    Replacing $S$ by $V\setminus S$ if necessary, we may assume that $r\notin S$, so that $S$ is an interval contained in $V\setminus\{r\}$.
\end{proof}

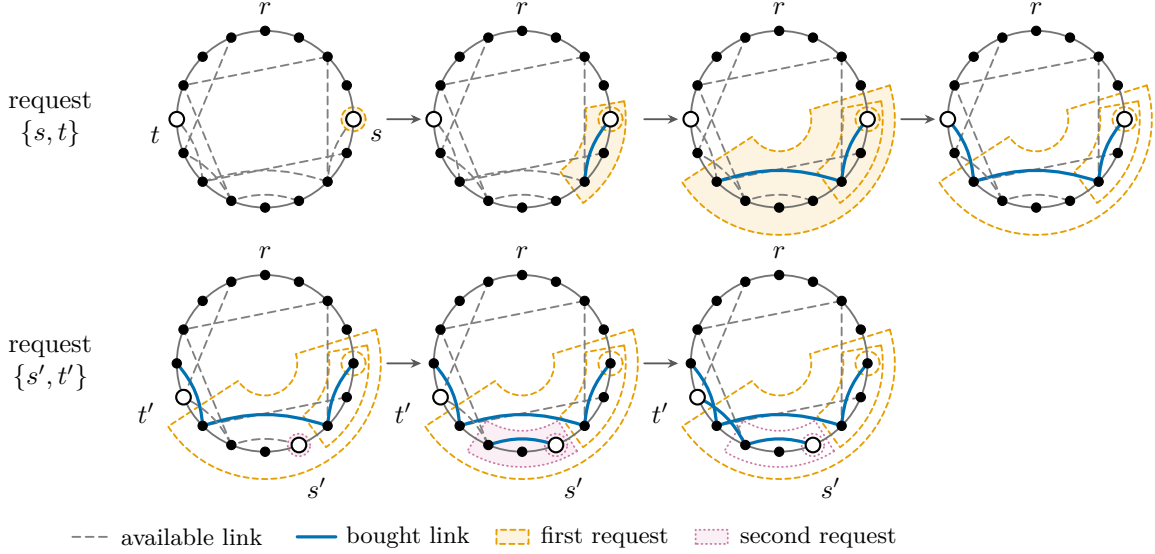
\begin{figure}[tbp]
    \centering
    \begin{tikzpicture}[paper/diagram,scale=.85]
    \foreach \row/\stage/\shift/\height in {
        0/0/3/0,0/1/7/0,0/2/11/0,0/3/15/0,
        1/0/3/-3.8,1/1/7/-3.8,1/2/11/-3.8%
    } {
        \begin{scope}[xshift=\shift cm, yshift=\height cm, scale=1.25]
            \foreach \i in {1,...,16} {
                \coordinate (v\i) at ({90-22.5*(\i-1)}:1.10);
            }
            \pgfmathtruncatemacro{\progress}{3*\row+\stage}
            \pgfmathtruncatemacro{\firststage}{min(3,\progress)}

            \foreach \level/\start/\finish/\inner/\outer in {
                2/17/-147/0.40/1.44,1/10/-54/0.80/1.29%
            } {
                \ifnum\firststage<\level\relax
                \else
                    \ifnum\firststage=\level\relax
                        \def\intervalfill{paperInterval!12}
                    \else
                        \def\intervalfill{none}
                    \fi
                    \path[paper/interval, fill=\intervalfill] (\start:\outer)
                        arc[start angle=\start, end angle=\finish, radius=\outer]
                        -- (\finish:\inner)
                        arc[start angle=\finish, end angle=\start, radius=\inner]
                        -- cycle;
                \fi
            }
            \ifnum\firststage=0
                \path[paper/interval] (v5) circle[radius=0.15];
            \else
                \path[paper/interval, paper/past interval] (v5) circle[radius=0.15];
            \fi

            \ifnum\row=1
                \ifnum\stage>0
                    \ifnum\stage=1
                        \def\intervalfill{paperSecondInterval!12}
                    \else
                        \def\intervalfill{none}
                    \fi
                    \path[paper/second interval, fill=\intervalfill] (-58:1.29)
                        arc[start angle=-58, end angle=-122, radius=1.29]
                        -- (-122:0.84)
                        arc[start angle=-122, end angle=-58, radius=0.84] -- cycle;
                \fi
                \ifnum\stage=0
                    \path[paper/second interval] (v8) circle[radius=0.14];
                \else
                    \path[paper/second interval, paper/past interval] (v8) circle[radius=0.14];
                \fi
            \fi

            \draw[paper/edge] (0,0) circle[radius=1.10];
            \foreach \u/\v in {3/7,3/14,6/11,10/14,12/16} {
                \draw[paper/link] (v\u) -- (v\v);
            }
            \foreach \purchase/\u/\v in {1/5/7,2/7/11,3/11/13,4/8/10,5/10/12} {
                \ifnum\progress<\purchase\relax
                    \draw[paper/link] (v\u) to[bend right=18] (v\v);
                \else
                    \draw[paper/bought link] (v\u) to[bend right=18] (v\v);
                \fi
            }

            \foreach \i in {1,...,16} {
                \node[paper/vertex] at (v\i) {};
            }
            \node[paper/label] at (0,1.38) {$r$};
            \ifnum\row=0
                \foreach \i in {5,13} {
                    \node[paper/terminal] at (v\i) {};
                }
                \ifnum\stage=0
                    \node[paper/label] at (-8:1.40) {$s$};
                    \node[paper/label] at (-172:1.38) {$t$};
                \fi
            \else
                \foreach \i in {8,12} {
                    \node[paper/terminal] at (v\i) {};
                }
                \node[paper/label] at (-67.5:1.68) {$s'$};
                \node[paper/label] at (-157.5:1.60) {$t'$};
            \fi
        \end{scope}
    }

    \foreach \shift/\height in {3/0,7/0,11/0,3/-3.8,7/-3.8} {
        \draw[paper/transition] ({\shift+1.90},\height) -- ({\shift+2.43},\height);
    }

    \node[paper/label, anchor=east, align=center] at (0.35,0) {request\\$\{s,t\}$};
    \node[paper/label, anchor=east, align=center] at (0.35,-3.8) {request\\$\{s',t'\}$};

    \draw[paper/link] (0.0,-6.50) -- (0.6,-6.50);
    \node[paper/annotation, anchor=west] at (0.72,-6.50) {available link};
    \draw[paper/bought link] (3.5,-6.50) -- (4.1,-6.50);
    \node[paper/annotation, anchor=west] at (4.22,-6.50) {bought link};
    \path[paper/interval] (6.6,-6.61) rectangle (7.1,-6.39);
    \node[paper/annotation, anchor=west] at (7.25,-6.50) {first request};
    \path[paper/second interval] (9.7,-6.61) rectangle (10.2,-6.39);
    \node[paper/annotation, anchor=west] at (10.35,-6.50) {second request};
\end{tikzpicture}
    \caption{Two successive requests handled by \Cref{alg:weighted}. Shading marks the exposed interval that is covered next. The exposed intervals of both requests, outlined in orange and pink, form a laminar family.}
    \label{fig:wcycap-request}
\end{figure}

\Cref{lem:interval-covering} suggests the following approach to WCycAP, in which all intervals are subsets of $V\setminus\{r\}$ in the linear order $v_2\prec\cdots\prec v_n$.
For a request $\{s,t\}$, as long as some interval $I \subseteq V\setminus\{r\}$ with $|\{s,t\} \cap I| = 1$ is not covered by the current solution, choose an inclusion-wise minimal such interval, which we call \emph{exposed}, and cover it.
\Cref{fig:wcycap-request} illustrates this process for two successive requests.
The exposed intervals form a \emph{laminar} family, that is, every two of them are disjoint or nested, even across different requests (\Cref{lem:exposed-laminar}).
We exploit this structure to split the exposed intervals into chains whose intervals contain a common vertex, where a family of intervals is a \emph{chain} if $I\subseteq J$ or $J\subseteq I$ for every two of its intervals $I$ and~$J$.
Each of these chains is served by an online algorithm for the following problem.

\paragraph{Online rooted chain covering.}
In online \emph{rooted chain covering}, we are given upfront a set of $n$ vertices $V = \{v_1,\ldots,v_n\}$ with a linear order $\prec$, a set $L$ of links on $V$ with nonnegative costs, and a vertex $p \in V$, which we call the \emph{root}.
The algorithm maintains a current solution $F \subseteq L$ of bought links, which initially consists of the zero-cost links, and requests arrive one by one and must be served immediately.
A request is an interval $I \subsetneq V$ with $p \in I$, and serving it means adding links to $F$ so that $F$ covers $I$.
The algorithm knows upfront that the requested intervals form a chain and that $L$ covers each of them, but they may arrive in any order.
The objective is to minimize the total cost of bought links.

\paragraph{The algorithm.}
Since the exposed intervals are revealed only online, we cannot decompose their laminar family into chains upfront.
Instead, \Cref{alg:weighted} assigns every exposed interval to a vertex of a binary search tree that is fixed upfront.
Let $T$ be a balanced binary search tree on $V\setminus\{r\}$ with respect to $\prec$, which we construct by recursively choosing medians.
Since $T$ has $n-1$ vertices, every root-to-leaf path in $T$ has at most $\lceil\log_2 n\rceil$ vertices.
For an interval $I\subseteq V\setminus\{r\}$, let $p(I)$ denote the vertex of $I$ of minimum depth in $T$.
This vertex is unique, since the lowest common ancestor of two distinct vertices of $I$ of equal depth has smaller depth and lies between them with respect to $\prec$, and thus in $I$.
We call $p(I)$ the \emph{pivot} of $I$.
Moreover, $p(I)$ is an ancestor in~$T$ of every other vertex $x\in I$.
Indeed, the lowest common ancestor of $p(I)$ and~$x$ lies between them with respect to~$\prec$, and thus in~$I$, and its depth is at most that of~$p(I)$, so it equals~$p(I)$ by the uniqueness of the pivot.

For every vertex $p \in V\setminus\{r\}$, \Cref{alg:weighted} maintains an instance $\cC_p$ of online rooted chain covering with vertices $V$, order $\prec$, links $L$, and root $p$, which we call a \emph{subinstance}.
Every subinstance runs its own copy of a fixed online algorithm for rooted chain covering, and all subinstances persist throughout the entire request sequence.
When a request $\{s,t\}$ arrives, \Cref{alg:weighted} repeatedly chooses an exposed interval $I$, sends it to the subinstance $\cC_{p(I)}$ of its pivot, and adds the links bought there to its solution, until $s$ and $t$ are $3$-edge-connected.
By \Cref{lem:interval-covering}, an exposed interval exists as long as the request is not served.
\Cref{fig:pivot-chains} shows how the exposed intervals are distributed among the subinstances.

\begin{figure}[tbp]
    \centering
    \begin{tikzpicture}[paper/diagram, x=0.7cm, y=0.85cm]
    \foreach \i in {1,...,16} {
        \coordinate (v\i) at (\i,0);
    }

    \foreach \i/\d in {9/1,5/2,13/2,3/3,7/3,11/3,15/3,2/4,4/4,6/4,8/4,10/4,12/4,14/4,16/4} {
        \coordinate (t\i) at (\i,{0.5+0.36*(4-\d)});
        \draw[paper/boundary guide, draw=black!18] (\i,0.08) -- (t\i);
    }
    \foreach \a/\b in {9/5,9/13,5/3,5/7,13/11,13/15,3/2,3/4,7/6,7/8,11/10,11/12,15/14,15/16} {
        \draw[paper/edge, draw=black!45] (t\a) -- (t\b);
    }
    \foreach \i in {2,...,16} {
        \node[paper/vertex, fill=white, line width=0.6pt, minimum size=4pt] at (t\i) {};
    }
    \node[paper/label, anchor=south] at ([yshift=4pt]t9) {$T$};

    \draw[paper/edge] (v1) -- (v16);
    \foreach \i in {1,...,16} {
        \node[paper/vertex] at (v\i) {};
    }
    \node[paper/label, anchor=base, inner sep=0pt] at (1,-0.42) {$r$};

    \foreach \i/\name/\colour in {9/p_2/black!70,5/p_1/col2,13/p_3/col3,15/p_4/col5} {
        \node[paper/vertex, draw=\colour, fill=\colour, minimum size=5.5pt] at (t\i) {};
        \node[paper/vertex, draw=\colour, fill=\colour, minimum size=5pt] at (v\i) {};
        \node[paper/label, text=\colour, anchor=base, inner sep=0pt] at (\i,-0.42) {$\name$};
        \draw[paper/pivot guide, draw=\colour] (\i,-0.55) -- (\i,-1.80);
        \node[paper/label, text=\colour, anchor=north] at (\i,-1.86) {$\cC_{\name}$};
    }

    \foreach \left/\right/\height/\colour in {
        8/10/-0.78/black!70, 7/11/-1.15/black!70, 3/14/-1.52/black!70,
        4/6/-0.78/col2, 3/6/-1.15/col2,
        12/13/-0.78/col3, 12/14/-1.15/col3,
        15/16/-0.78/col5%
    } {
        \draw[paper/interval bar, draw=\colour]
            ({\left-0.28},\height) -- ({\right+0.28},\height)
            ({\left-0.28},{\height-0.07}) -- ({\left-0.28},{\height+0.07})
            ({\right+0.28},{\height-0.07}) -- ({\right+0.28},{\height+0.07});
    }
\end{tikzpicture}
    \caption{The balanced binary search tree $T$ on $V\setminus\{r\}$ and a laminar family of exposed intervals.
    Every interval has the color of its pivot (its vertex of minimum depth in~$T$) and is sent to the subinstance of this pivot.
    The intervals with a common pivot form a chain.}
    \label{fig:pivot-chains}
\end{figure}
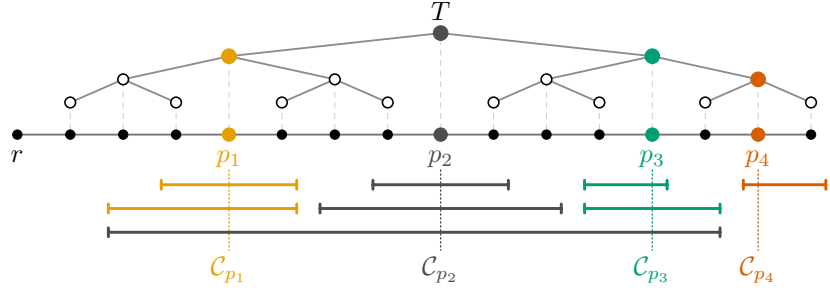

\begingroup
\LinesNotNumbered
\DontPrintSemicolon
\begin{algorithm}[t]
    \caption{Algorithm for online WCycAP}\label{alg:weighted}
    \SetKwBlock{OnRequest}{Upon arrival of request $\{s,t\}$}{}
    Build a balanced binary search tree $T$ on $V\setminus\{r\}$ with respect to $\prec$.\;
    Initialize $F=\{\ell \in L \mid c_\ell = 0\}$.\;
    \ForEach{$p\in V\setminus\{r\}$}{
        Start an instance $\cC_p$ of online rooted chain covering with vertices $V$, order $\prec$, links $L$, and root $p$.\;
        Initialize a copy of the algorithm for rooted chain covering on $\cC_p$ that maintains its own solution $F_p$.\;
    }
    \OnRequest{
        \While{$s$ and $t$ are not $3$-edge-connected in $C+F$}{
            Choose an $\subseteq$-wise minimal interval $I\subseteq V\setminus\{r\}$ s.t.\ $|\{s,t\}\cap I|=1$ and $\delta_F(I)=\emptyset$.\;
            Send $I$ to $\cC_{p(I)}$ and let $P$ be the set of links newly bought in $\cC_{p(I)}$ to serve $I$.\;
            Update $F\gets F\cup P$.\;
        }
    }
\end{algorithm}
\endgroup

\paragraph{Correctness.}
The proof of the following lemma adapts the uncrossing argument for minimal violated sets in network design~\cite{WilliamsonGMV95,GabowGW98} to exposed intervals of different requests.

\begin{lemma}\label{lem:exposed-laminar}
    The intervals that \Cref{alg:weighted} exposes form a laminar family.
\end{lemma}

\begin{proof}
    Let $I$ and $J$ be distinct exposed intervals, where $I$ was exposed before~$J$, and let $F$ be the set of bought links just before $I$ is exposed.
    Then $F$ does not cover $I$, and it does not cover $J$ either, since $J$ is exposed later, when the set of bought links is a superset of~$F$.

    For the sake of contradiction, suppose that $I$ and $J$ intersect and neither contains the other. Since both are intervals in the linear order on $V\setminus\{r\}$, the sets $I\cap J$ and $I\setminus J$ partition $I$ into two nonempty proper subintervals that also avoid $r$.
    Both $I \cap J$ and $I \setminus J$ are also uncovered by $F$, since every link that covers one of them covers $I$ or $J$.
    When the algorithm exposed $I$, the interval $I$ separated the two vertices of the current request $\{s,t\}$, say $s \in I$ and $t \notin I$.
    Exactly one of $I \cap J$ and $I \setminus J$ contains $s$ and neither contains $t$, so this interval is a strictly smaller uncovered interval that separates $s$ and $t$, contradicting the inclusion-wise minimality of $I$.
\end{proof}

\begin{lemma}\label{lem:exposed-chain}
    The exposed intervals sent to the subinstance $\cC_p$ form a chain, contain $p$, and are covered by~$L$. Thus, they are a valid input to $\cC_p$.
\end{lemma}

\begin{proof}
    By definition, every exposed interval sent to $\cC_p$ has pivot~$p$ and thus contains~$p$, so these intervals pairwise intersect. Hence, they form a chain by \Cref{lem:exposed-laminar}. Moreover, they avoid $r$ and are thus proper subsets of~$V$.
    Finally, every exposed interval separates the two vertices of its request, and $L$ serves this request, so $L$ covers the interval by \Cref{lem:interval-covering}.
\end{proof}

Finally, we argue that \Cref{alg:weighted} serves every request.
Consider an iteration of the while loop for a request $\{s,t\}$ with exposed interval $I$.
By \Cref{lem:exposed-chain}, $I$ is a valid request for $\cC_{p(I)}$.
Afterwards, the solution of $\cC_{p(I)}$ covers $I$, and all of its links have cost zero or were added to $F$, so $F$ covers $I$.
Since $F$ only grows, every iteration increases the set of covered intervals, so the while loop terminates, and then $s$ and $t$ are $3$-edge-connected.

It remains to choose the rooted chain covering algorithm that the subinstances run and to bound the cost of the resulting solution.

\subsection{A Primal-Dual Algorithm for Rooted Chain Covering}\label{sec:interval-chain}

In this subsection, we consider online rooted chain covering on its own and give a deterministic $O(\log n)$-competitive primal-dual algorithm for it.

Online rooted chain covering generalizes online WPAP.
Indeed, the minimum cuts of the path $v_1 - v_2 - \cdots - v_n$ are the suffixes $\{v_i,\ldots,v_n\}$ with $2\le i\le n$ and their complements.
Hence, by \Cref{lem:cut-covering}, for $p=v_n$, a WPAP request $\{v_a,v_b\}$ with $a<b$ can be simulated by requesting the suffixes for $i=a+1,\ldots,b$ one after another, which preserves feasible solutions and optimal costs.
Therefore, the lower bound of Naor, Umboh, and Williamson~\cite{NaorUW22} for WPAP shows that every deterministic algorithm for online rooted chain covering has competitive ratio $\Omega(\log n)$, so \Cref{thm:interval-chain-augmentation} is tight up to a constant factor.
For general~$p$, the requested intervals, ordered by inclusion, still form a path in which every link covers a subpath, since a link that covers two requested intervals $I\subseteq J$ has one endpoint in~$I$ and the other outside~$J$.
However, the intervals that contain~$p$ need not be nested, so unlike in WPAP, this path is revealed only online.

Throughout this subsection, let $\cI$ denote the set of requested intervals of the instance.
Whenever a requested interval arrives a second time, the current solution already covers it, so we may assume that all requested intervals are distinct.
Since a chain of distinct proper intervals that contain $p$ has at most $n-1$ members, we have $|\cI| <n$.
Recall that every link has cost $0$ or a power of $2$, and that a link of cost $2^j$ belongs to class~$j$.
In particular, all links of a class have the same cost.

Our algorithm is based on the following LP relaxation and its dual.
In the primal, a variable $x_\ell$ for every link $\ell \in L$ indicates whether $\ell$ is bought. 
The primal constraints ensure that every requested interval $I \in \cI$ is covered by at least one bought link.

\vspace{-0.5cm}
\begin{minipage}[t]{0.48\textwidth}
\centering 
\begin{alignat*}{3}
    \text{(LP)} \quad \text{minimize} \quad
    &\sum_{\ell \in L} c_\ell x_\ell \\
    \text{subject to} \quad
    &\sum_{\mathclap{\ell \in \delta_L(I)}} x_\ell \ge 1
        &\quad& \forall I \in \cI \\
        &x_\ell \geq0 && \forall \ell \in L
\end{alignat*}
\end{minipage}\hfill
\hfill
\begin{minipage}[t]{0.48\textwidth}
\centering 
\begin{alignat*}{3}
    \text{(DLP)} \quad \text{maximize} \quad
    &\sum_{I \in \cI} y_I \\
    \text{subject to} \quad
    &\sum_{\mathclap{I: \ell \in \delta_L(I)}} y_I \le c_\ell
        &\quad& \forall \ell \in L \\
        &y_I \geq0 && \forall I \in \cI
\end{alignat*}
\end{minipage}

\bigskip

Our algorithm increases dual variables.
When an uncovered interval $I$ arrives, it raises the dual variable $y_I$ until the dual constraint of some link covering $I$ becomes tight.
We call this link the \emph{trigger} of $I$.
For WPAP, Naor, Umboh, and Williamson~\cite{NaorUW22} buy the trigger after pruning the links relative to the known path such that every path edge is covered by at most two links per class.
Since our path of requested intervals is revealed only online, our algorithm does not prune and does not necessarily buy the trigger.
Instead, it buys at most four \emph{representatives} of the trigger's class (\Cref{lem:nested-covers}), which cover every requested interval that is covered by some link of this class covering~$I$.

\begin{lemma}[Representative Lemma]\label{lem:nested-covers}
    Let $I \subseteq V$ be an interval with $p \in I$ and let $A \subseteq \delta_L(I)$.
    There exists a set $D \subseteq A$ with $|D| \le 4$ such that every interval $J\subseteq V$ with $p\in J$ and with $I\subseteq J$ or $J\subseteq I$ that is covered by $A$ is also covered by~$D$.
\end{lemma}
\begin{figure}[tbp]
    \centering
    \begin{tikzpicture}[paper/diagram, y=0.85cm]
    \foreach \panel/\shift in {0/0,1/8} {
        \begin{scope}[xshift=\shift cm]
            \foreach \i in {1,...,21} {
                \coordinate (v\i) at ({0.32*(\i-1)},0);
            }
            \ifnum\panel=0
                \node[paper/label] at (3.2,-1.70) {(a) $J\subseteq I$};
                \def\jleft{9}\def\jright{13}
                \def\ilevel{-1.12}\def\jlevel{-0.72}
            \else
                \node[paper/label] at (3.2,-1.70) {(b) $I\subseteq J$};
                \def\jleft{4}\def\jright{18}
                \def\ilevel{-0.72}\def\jlevel{-1.12}
            \fi

            \draw[paper/edge] (v1) -- (v21);
            \foreach \u/\v/\height in {3/8/0.65,5/9/0.45,13/17/0.45,14/19/0.65} {
                \draw[paper/link] (v\u)
                    .. controls +(0,\height) and +(0,\height) .. (v\v);
            }
            \foreach \u/\v/\name/\case/\height in {
                4/10/d_1/0/1.20,12/18/d_2/0/1.20,
                1/7/d_3/1/1.95,15/21/d_4/1/1.95%
            } {
                \ifnum\panel=\case\relax
                    \def\linkcolour{paperBought}
                    \def\linkwidth{1.2pt}
                \else
                    \def\linkcolour{paperBought!40}
                    \def\linkwidth{0.7pt}
                \fi
                \draw[paper/bought link, draw=\linkcolour, line width=\linkwidth] (v\u)
                    .. controls +(0,\height) and +(0,\height) ..
                    node[paper/annotation, text=\linkcolour, midway, above=3pt] {$\name$} (v\v);
            }
            \foreach \i in {1,...,21} {
                \node[paper/vertex] at (v\i) {};
            }
            \node[paper/terminal] at (v11) {};
            \node[paper/label, anchor=base, inner sep=0pt] at (3.2,-0.42) {$p$};
            \draw[paper/pivot guide, draw=black!35] (3.2,-0.55) -- (3.2,-1.12);
            \ifnum\panel=0
                \foreach \i in {10,12} {
                    \node[paper/vertex, draw=paperBought, fill=paperBought, minimum size=5pt] at (v\i) {};
                }
            \else
                \foreach \i in {1,21} {
                    \node[paper/vertex, draw=paperBought, fill=paperBought, minimum size=5pt] at (v\i) {};
                }
            \fi

            \foreach \left/\right/\height/\colour/\name in {
                6/16/\ilevel/paperNeutral/I,\jleft/\jright/\jlevel/col2/J%
            } {
                \draw[paper/interval bar, draw=\colour]
                    ({0.32*(\left-1)-0.12},\height) -- ({0.32*(\right-1)+0.12},\height);
                \foreach \endpoint in {{0.32*(\left-1)-0.12},{0.32*(\right-1)+0.12}} {
                    \draw[paper/boundary guide] ({\endpoint},-0.08) -- ({\endpoint},\height);
                    \draw[paper/interval bar, draw=\colour]
                        ({\endpoint},{\height-0.07}) -- ({\endpoint},{\height+0.07});
                }
                \node[paper/label, text=\colour, anchor=west]
                    at ({0.32*(\right-1)+0.24},\height) {$\name$};
            }
        \end{scope}
    }
\end{tikzpicture}
    \caption{The representatives of \Cref{lem:nested-covers} for an interval $I$ with $p\in I$ and a set $A\subseteq\delta_L(I)$.
    (a) For $J\subseteq I$, the links $d_1$ and $d_2$, whose endpoints in $I$ are closest to $p$, cover $J$ whenever $A$ does.
    (b) For $I\subseteq J$, the links $d_3$ and $d_4$, whose endpoints outside $I$ are farthest from $p$, cover $J$ whenever $A$ does.}
    \label{fig:nested-covers}
\end{figure}
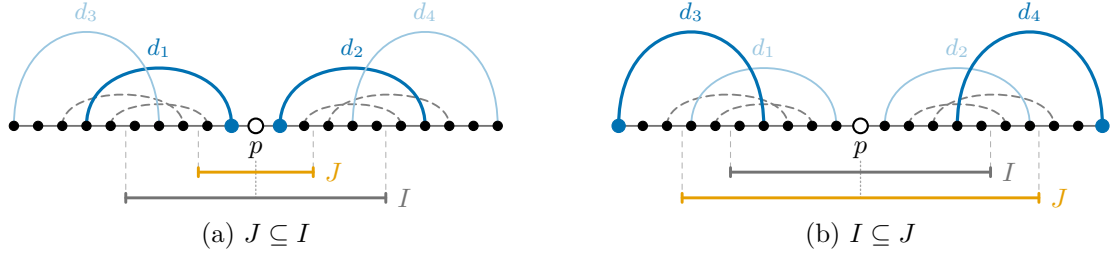
\begin{proof}
    If $A=\emptyset$, then $D=\emptyset$ satisfies the statement, so assume that $A \neq \emptyset$.
    We define the following at most four links.
    \begin{itemize}
        \item 
        Let $D_1$ be the set of all links $uv \in A$ whose endpoint $v \in I$ satisfies $v \preceq p$.
        If $D_1 \neq \emptyset$, let $d_1 \in D_1$ be a link whose endpoint in $I$ is closest to $p$ among all links in $D_1$.
        \item 
        Let $D_2$ be the set of all links $uv \in A$ whose endpoint $v \in I$ satisfies $p \preceq v$.
        If $D_2 \neq \emptyset$, let $d_2 \in D_2$ be a link whose endpoint in $I$ is closest to $p$ among all links in $D_2$.
        \item 
        Let $D_3$ be the set of all links $uv \in A$ whose endpoint $u \notin I$ satisfies $u \preceq p$.
        If $D_3 \neq \emptyset$, let $d_3 \in D_3$ be a link whose endpoint in $V \setminus I$ is farthest from $p$ among all links in $D_3$.
        \item 
        Let $D_4$ be the set of all links $uv \in A$ whose endpoint $u \notin I$ satisfies $p \preceq u$.
        If $D_4 \neq \emptyset$, let $d_4 \in D_4$ be a link whose endpoint in $V \setminus I$ is farthest from $p$ among all links in $D_4$.
    \end{itemize}   
    Since $A\neq\emptyset$, at least one of $d_1$ and $d_2$ exists, and at least one of $d_3$ and $d_4$ exists.
    Let $D$ be the set of those links among $d_1,d_2,d_3,d_4$ that exist. Thus, $1\le|D|\le4$.
    \Cref{fig:nested-covers} illustrates these representatives and the two cases below.

    We now show that $D$ satisfies the stated property.
    Let $J$ be an interval with $p \in J$ such that $I \subseteq J$ or $J \subseteq I$.
    Let $\ell \in A$ be a link that covers $J$. 
    Since $\ell \in A \subseteq \delta_L(I)$, it also covers $I$. Let $\ell=uv$ with $v \in I$ and $u \notin I$.
    We distinguish between two cases.
    
    \textbf{Case $J \subseteq I$.}
    Since $u \notin I \supseteq J$ and $\ell$ covers $J$, we have $v \in J$.
    By symmetry, we may assume that $v \preceq p$. Then $\ell \in D_1$, so $d_1$ exists, say $d_1 = a_1b_1$ with $a_1 \in I$ and $b_1 \notin I$.
    By the choice of $d_1$, we have $v \preceq a_1 \preceq p$.
    Since $J$ is an interval that contains $v$ and $p$, it also contains $a_1$. Moreover, $b_1 \notin I \supseteq J$. Hence, $d_1$ covers $J$.

    \textbf{Case $I \subseteq J$.}
    Then $v \in I \subseteq J$, and since $\ell$ covers $J$, we have $u \notin J$. 
    By symmetry, we may assume that $u \preceq p$. Then $\ell \in D_3$, so $d_3$ exists, say $d_3 = a_3b_3$ with $a_3 \in I$ and $b_3 \notin I$.
    By the choice of $d_3$, we have $b_3 \preceq u \preceq p$.
    Since $J$ is an interval that contains $p$ but not $u$, it does not contain $b_3$.
    As $a_3 \in I \subseteq J$, the link $d_3$ covers $J$.
\end{proof}

\begingroup
\LinesNotNumbered
\DontPrintSemicolon
\begin{algorithm}[t]
    \caption{Primal-dual algorithm for online rooted chain covering}\label{alg:interval-chain}
    \SetKwBlock{OnRequest}{Upon arrival of interval $I$}{}
    Initialize $F=\{\ell \in L \mid c_\ell = 0\}$ and $y_J=0$ for every interval $J\subsetneq V$ with $p\in J$.\;
    \OnRequest{
        \If{$I$ is not covered by $F$}{
            Continuously raise $y_I$, maintaining the dual constraints $\sum_{J : \ell \in \delta_L(J)} y_J \le c_\ell$ for all $\ell \in L$,
            until a constraint for a link in $\delta_L(I)$ is tight.\;
            Choose any such link $\ell$ as the trigger of $I$, and let $j$ be its class, that is, $c_\ell=2^j$.\;
            Apply \Cref{lem:nested-covers} to $A_j(I)=\{\ell'\in\delta_L(I)\mid c_{\ell'} = 2^j\}$ to obtain $D\subseteq A_j(I)$.\;
            Buy all links in $D$, that is, $F\gets F\cup D$.\;
        }
    }
\end{algorithm}
\endgroup

\Cref{alg:interval-chain} maintains a dual solution $y$ with $y_I \ge 0$ for all requests $I$ and a set $F$ of bought links.
When an uncovered interval~$I$ arrives, we have $\delta_L(I)\neq\emptyset$, since $L$ covers~$I$, so the dual constraint of some link in $\delta_L(I)$ eventually becomes tight. It may already be tight when $y_I=0$, in which case we choose this link without increasing $y_I$.
Since $F$ contains all zero-cost links and does not cover $I$, every link in $\delta_L(I)$ has positive cost, so the trigger belongs to a class.
The trigger $\ell$ belongs to $A_j(I)$, so $A_j(I)$ covers $I$, and \Cref{lem:nested-covers} with $J=I$ shows that $D$ also covers $I$.
Hence, \Cref{alg:interval-chain} serves every request.

\paragraph{Competitive analysis.}
The following lemma captures the effect of buying representatives and bounds the number of triggers of each class that cover a fixed requested interval.

\begin{lemma}\label{lem:few-triggers-per-class}
    For every class~$j$, every link of class~$j$ covers at most one request whose trigger belongs to class~$j$, and for every requested interval $K \in \cI$, at most two triggers of class~$j$ cover~$K$.
\end{lemma}

\begin{proof}
    For the first part, suppose that a link $\ell$ of class $j$ covers two requests $I \neq I'$ whose triggers belong to class~$j$, where $I$ arrives before $I'$. Then $\ell \in A_j(I)$, so $A_j(I)$ covers $I'$. Since $I \subseteq I'$ or $I' \subseteq I$, and $p \in I'$, \Cref{lem:nested-covers} guarantees that the set $D$ bought for request $I$ also covers $I'$. Hence, $I'$ is already covered when it arrives and has no trigger, a contradiction.

    For the second part, suppose for the sake of contradiction that three distinct triggers of class $j$ cover $K$.
    Their requests are distinct, as each request has at most one trigger.
    Since the requested intervals form a chain, each of these three requests is contained in $K$ or contains $K$.
    Hence, two of them are both contained in $K$ or both contain $K$, say intervals $I, I'$ such that $I \subseteq I' \subseteq K$ or $K \subseteq I' \subseteq I$.
    The trigger of $I$ covers $I$ and $K$, so it has one endpoint in the smaller of these two intervals and the other outside the larger one.
    Thus, in both cases, it also covers $I'$, which lies between them, contradicting the first part.
\end{proof}

\begin{theorem}\label{thm:interval-chain-augmentation}
    \Cref{alg:interval-chain} is a deterministic algorithm for online rooted chain covering whose competitive ratio on instances with $n$ vertices is at most $8 \lceil \log_2 (2n) \rceil + 10$.
\end{theorem}

\begin{proof}
    We charge the cost of the algorithm to the triggers and the cost of the triggers to the final dual objective.
    Triggers that are much cheaper than this objective cost little in total, and the others belong to $O(\log n)$ classes, so \Cref{lem:few-triggers-per-class} bounds their total cost by $O(\log n)$ times this objective.

    Let $\cT \subseteq L$ denote the set of all triggers chosen throughout the instance.
    Since every trigger covers its request, \Cref{lem:few-triggers-per-class} shows that every trigger $\ell \in \cT$ is the trigger of exactly one request.
    For this request, the algorithm buys at most four links of the same class as $\ell$ by \Cref{lem:nested-covers}, at a total cost of at most $4 c_\ell$.
    Thus, it suffices to bound $\sum_{\ell\in\cT} c_\ell$.

    Fix the final dual solution $y$ and let $Y = \sum_{I \in \cI} y_I$ be its objective value.
    The dual constraint of every trigger is tight when the trigger is chosen, and it remains tight until the end of the instance, since the dual variables never decrease and \Cref{alg:interval-chain} maintains all dual constraints.
    Hence, every trigger $\ell \in \cT$ satisfies
    \[
        c_\ell = \sum_{I : \ell \in \delta_{L}(I)} y_I \le Y \ .
    \]

    Let $\cT_{\ge}$ denote the set of triggers $\ell \in \cT$ with $c_\ell \ge Y/(2n)$.
    Since there are fewer than $n$ requests and each of them has at most one trigger, we have $|\cT| \le n$, so the remaining triggers have total cost $\sum_{\ell \in \cT \setminus \cT_{\ge}} c_\ell \le n\cdot Y/(2n) = Y/2$.
    Since the costs of triggers are powers of~$2$ and at most~$Y$, each trigger in $\cT_{\ge}$ belongs to one of the at most $\lceil \log_2 (2n) \rceil + 1$ classes whose cost lies between $Y/(2n)$ and~$Y$.
    By \Cref{lem:few-triggers-per-class}, for each requested interval $I \in \cI$, there are at most $2(\lceil \log_2 (2n) \rceil + 1)$ triggers in $\cT_\ge$ that cover $I$, that is, $|\delta_{\cT_\ge}(I)| \le 2(\lceil \log_2 (2n) \rceil + 1)$.
    Thus, summing the tight dual constraints of the links in $\cT_{\ge}$ gives
    \[
    \sum_{\ell \in \cT_{\ge}} c_\ell = \sum_{I\in\cI}  y_I |\delta_{\cT_\ge}(I)|
    \le 2(\lceil \log_2 (2n) \rceil + 1) \cdot Y  \ .
    \]
    Therefore, the total cost of \Cref{alg:interval-chain} satisfies
    \[
    c(F) \le 4 \biggl(\frac Y2 + \sum_{\ell \in \cT_{\ge}} c_\ell \biggr) \le 4 \biggl(\frac Y2 + 2(\lceil \log_2 (2n) \rceil + 1) \cdot Y \biggr) = (8\lceil \log_2 (2n) \rceil + 10) \cdot Y \ .
    \]
    Let $F^*$ be an optimal solution to the instance.
    Since $y$ is feasible for (DLP), weak duality gives $Y \le c(F^*)$.
    Hence, $c(F) \le (8\lceil \log_2 (2n) \rceil + 10)\cdot c(F^*)$.
\end{proof}

\subsection{Combining the Chains}\label{sec:log-squared}

We now return to WCycAP and bound how often a link of an optimal solution is charged across the chains. The two counting lemmas in this subsection depend only on the exposed intervals and their pivots, so they hold for any rooted chain covering algorithm used by \Cref{alg:weighted}.

Fix an instance and an optimal solution $F^*$ for it.
For $p \in V \setminus \{r\}$, we also use $\cC_p$ to denote the set of all intervals that \Cref{alg:weighted} sends to the subinstance $\cC_p$, and we let $F_p$ denote the final solution of this subinstance.
Let $F_p^*$ be the set of links of $F^*$ that cover some interval of~$\cC_p$.

Every interval of $\cC_p$ separates the two vertices of the request for which it was exposed. By \Cref{lem:interval-covering}, it is therefore covered by some link of $F^*$, which belongs to $F^*_p$.
Hence, $F^*_p$ is a feasible solution to $\cC_p$.
The next lemma bounds the number of these sets that contain a given link.

\begin{lemma}\label{lem:active-rooted-count}
    Every link $\ell\in F^*$ belongs to $F^*_p$ for at most $2\lceil\log_2 n\rceil$ vertices $p\in V\setminus\{r\}$.
\end{lemma}

\begin{proof}
Let $\ell=uv\in F^*$, and let $p\in V\setminus\{r\}$ with $\ell\in F^*_p$.
Then $\ell$ covers some interval $I \in \cC_p$, so $I$ contains exactly one endpoint $x\in\{u,v\}$ of~$\ell$, and $x\neq r$, as $I\subseteq V\setminus\{r\}$.
Since $p=p(I)$ is an ancestor in~$T$ of every other vertex of~$I$, the vertex $p$ is $x$ or an ancestor of~$x$ in~$T$.
Every vertex of $V\setminus\{r\}$ has at most $\lceil \log_2 n \rceil$ ancestors in~$T$, counting itself, and $x$ is $u$ or~$v$, so there are at most $2\lceil \log_2 n \rceil$ vertices $p$ with $\ell \in F^*_p$.
\end{proof}

\paragraph{An $O(\log^2n)$ guarantee.}
If every subinstance runs \Cref{alg:interval-chain}, then \Cref{alg:weighted} is an $O(\log^2 n)$-competitive deterministic algorithm for online WCycAP. Indeed, let $\alpha\in O(\log n)$ be the competitive ratio from \Cref{thm:interval-chain-augmentation}. Since each $F_p^*$ is feasible for $\cC_p$, every positive-cost link of~$F$ is bought in a subinstance, and \Cref{lem:active-rooted-count} bounds the multiplicity of each optimal link, we have
\begin{equation*}
c(F)\le\sum_{p\in V\setminus\{r\}}c(F_p)\le\alpha\sum_{p\in V\setminus\{r\}}c(F_p^*)\le 2\lceil\log_2 n\rceil\cdot\alpha\cdot c(F^*) = O(\log^2 n) \cdot c(F^*) \ .
\end{equation*}
This recovers the bound of Garg and Subramanian~\cite{GargS26}. To improve it, we distinguish the chains on which a link covers the largest requested interval from those on which it does not.

\paragraph{The key to $O(\log n)$.}
For $p \in V \setminus \{r\}$, let $M_p$ denote the largest interval in $\cC_p$ if $\cC_p$ is nonempty, and set $M_p=\emptyset$ otherwise.
Every link of $F^*$ that covers $M_p$ belongs to~$F^*_p$.
Hence, $F_p^*$ is the disjoint union of the set $R^*_p = \delta_{F^*}(M_p)$ of links that cover $M_p$ and the set $S^*_p = F^*_p \setminus R^*_p$ of remaining links.
A link of $S^*_p$ covers an interval of $\cC_p$, which is contained in $M_p$, but it does not cover $M_p$. Hence, both of its endpoints lie in $M_p$.
This partition follows the approach of Naor, Umboh, and Williamson~\cite{NaorUW22} for online WTAP, which decomposes the tree into paths such that every link plays the role of a link of $S^*_p$ on at most one path and of $R^*_p$ on $O(\log n)$ paths.
Since the laminar family of exposed intervals is not known upfront, the pivots take the place of this decomposition in our setting.

\begin{figure}[tbp]
    \centering
    \begin{tikzpicture}[paper/diagram, x=0.42cm, y=0.85cm]
    \foreach \i in {1,...,20} {
        \coordinate (v\i) at (\i,0);
    }

    \fill[col6!7] (7.6,-0.12) rectangle (12.4,-2.22);
    \draw[paper/boundary guide, draw=col6!45] (7.6,-0.08) -- (7.6,-2.22);
    \draw[paper/boundary guide, draw=col6!45] (12.4,-0.08) -- (12.4,-2.22);

    \draw[paper/edge] (v1) -- (v20);
    \draw[paper/edge, dotted] (v1) -- ++(-1.2,0) (v20) -- ++(1.2,0);

    \draw[paper/internal optimal link] (v9)
        .. controls +(0,0.75) and +(0,0.75) .. (v11);
    \node[paper/annotation, text=col4] at (10,0.78) {$\ell_1$};
    \draw[paper/rooted optimal link] (v12)
        .. controls +(0,0.85) and +(0,0.85) .. (v14);
    \draw[paper/internal optimal link, paper/secondary optimal link] (v12)
        .. controls +(0,1.05) and +(0,1.05) .. (v14);
    \node[paper/annotation] at (13,1.03) {$\ell_2$};
    \draw[paper/rooted optimal link] (v1)
        .. controls +(0,1.00) and +(0,1.00) .. (v8);
    \draw[paper/rooted optimal link, paper/secondary optimal link] (v1)
        .. controls +(0,1.20) and +(0,1.20) .. (v8);
    \draw[paper/rooted optimal link, paper/tertiary optimal link] (v1)
        .. controls +(0,1.40) and +(0,1.40) .. (v8);
    \node[paper/annotation] at (4.5,1.28) {$\ell_3$};

    \foreach \i in {1,...,20} {
        \node[paper/vertex] at (v\i) {};
    }

    \foreach \i/\name/\colour/\bottom in {10/p/col6/-1.15, 15/{p'}/col3/-1.50, 5/{p''}/col2/-1.85} {
        \node[paper/vertex, draw=\colour, fill=\colour, minimum size=5pt] at (v\i) {};
        \node[paper/label, text=\colour, anchor=base, inner sep=0pt] at (\i,-0.42) {$\name$};
        \draw[paper/pivot guide, draw=\colour] (\i,-0.55) -- (\i,\bottom);
        \node[paper/label, text=\colour, anchor=north] at (\i,-2.32) {$\cC_{\name}$};
    }

    \foreach \left/\right/\height/\colour/\width in {
        10/11/-0.80/col6/1.1pt, 8/12/-1.15/col6/1.7pt,
        14/16/-0.80/col3/1.1pt, 7/17/-1.50/col3/1.7pt,
        3/5/-0.80/col2/1.1pt, 2/19/-1.85/col2/1.7pt%
    } {
        \draw[paper/interval bar, draw=\colour, line width=\width]
            ({\left-0.32},\height) -- ({\right+0.32},\height)
            ({\left-0.32},{\height-0.07}) -- ({\left-0.32},{\height+0.07})
            ({\right+0.32},{\height-0.07}) -- ({\right+0.32},{\height+0.07});
    }
    \node[paper/label, text=col6, anchor=west] at (12.75,-1.15) {$M_p$};
    \node[paper/label, text=col3, anchor=west] at (17.75,-1.50) {$M_{p'}$};
    \node[paper/label, text=col2, anchor=west] at (19.75,-1.85) {$M_{p''}$};

    \foreach \y/\style/\name in {0.75/{}/p, 0.20/{paper/secondary optimal link}/{p'}} {
        \draw[paper/internal optimal link, \style] (23.0,\y) -- (24.2,\y);
        \node[paper/annotation, anchor=west] at (24.44,\y) {$S^*_{\name}$};
    }
    \foreach \y/\style/\name in {0.75/{}/p, 0.20/{paper/secondary optimal link}/{p'}, -0.35/{paper/tertiary optimal link}/{p''}} {
        \draw[paper/rooted optimal link, \style] (26.6,\y) -- (27.8,\y);
        \node[paper/annotation, anchor=west] at (28.04,\y) {$R^*_{\name}$};
    }
\end{tikzpicture}
    \caption{Three subinstances with nested largest intervals $M_p\subsetneq M_{p'}\subsetneq M_{p''}$ and three links of $F^*$.
    Each link $\ell$ is drawn once for every subinstance $\cC_x$ with $\ell\in F^*_x$, in black if it belongs to $R^*_x$ and in pink if it belongs to $S^*_x$.
    By \Cref{lem:inner-disjoint}, every link has at most one pink arc.
    }
    \label{fig:nested-subinstances}
\end{figure}
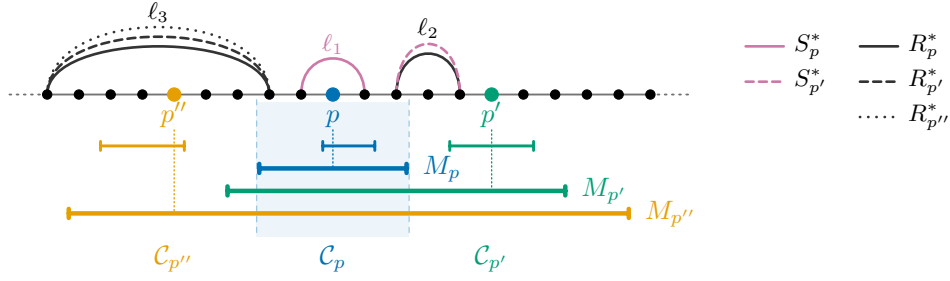

While a link $\ell \in F^*$ may belong to $R^*_p$ for up to $2\lceil\log_2 n\rceil$ vertices $p$ by \Cref{lem:active-rooted-count}, the next lemma shows that $\ell$ belongs to $S^*_p$ for at most one vertex $p$.
\Cref{fig:nested-subinstances} illustrates both kinds of links.
In this figure, every interval of $\cC_{p'}$ and $\cC_{p''}$ is disjoint from the shaded interval $M_p$ or contains it, so $\ell_1$, which has both endpoints in $M_p$, belongs to $F^*_x$ only for $x=p$.

\begin{lemma}\label{lem:inner-disjoint}
    The sets $(S^*_p)_{p \in V \setminus \{r\}}$ are pairwise disjoint.
\end{lemma}

\begin{proof}
For the sake of contradiction, suppose that there exists a link $\ell \in S^*_p \cap S^*_{p'}$ with $p \neq p'$.
Then both endpoints of $\ell$ lie in $M_p \cap M_{p'}$, and $M_p \neq M_{p'}$, as otherwise $p = p(M_p) = p(M_{p'}) = p'$.
Since $M_p$ and $M_{p'}$ are distinct exposed intervals that intersect, \Cref{lem:exposed-laminar} implies that they are nested. By symmetry, we may assume that $M_p \subsetneq M_{p'}$.
Moreover, $p' \notin M_p$, since otherwise $p'$ would have minimum depth among the vertices of $M_{p'}\supseteq M_p$, and thus $p' = p(M_{p}) = p$.

Since $\ell \in S_{p'}^* \subseteq F_{p'}^*$, there exists an interval $J \in \cC_{p'}$ that is covered by $\ell$.
Both endpoints of $\ell$ lie in $M_p$, and so does its endpoint in $J$, so $M_p$ and $J$ intersect.
Moreover, $p' \in J$ and $p' \notin M_p$, so $J\not\subseteq M_p$. Hence, $M_p\subseteq J$ by \Cref{lem:exposed-laminar}.
But then both endpoints of $\ell$ lie in $J$, which contradicts that $\ell$ covers $J$.
\end{proof}

Consequently, a per-chain bound $c(F_p)\le O(c(R_p^*)+\log n\cdot c(S_p^*))$ would give an $O(\log n)$-competitive algorithm. Under this bound, each optimal link is charged a constant factor on at most $2\lceil\log_2 n\rceil$ chains and a logarithmic factor on at most one. The next subsection obtains this bound without knowing the largest intervals upfront.

\subsection{A Black-Box Conversion for Rooted Chain Covering}\label{sec:refined}

We now consider online rooted chain covering on its own, with root~$p\in V$, and let $F^*$ be any feasible solution, which need not be optimal.
Let $M_p$ be the largest requested interval, let $R^*=\delta_{F^*}(M_p)$ be the set of links of $F^*$ that cover it, and let $S^*=F^*\setminus R^*$.
A link $uv\in R^*$ with $u\in M_p$ covers a requested interval $J$ if and only if $u\in J$.
Since the requested intervals form a chain, a single link $\ell^* \in R^*$ covers every requested interval that $R^*$ covers, namely a link that covers the smallest such interval, and $S^*$ covers all other requested intervals.
The cost of $\ell^*$ is at most $c(R^*)$, but an online algorithm does not know this link upfront.

Naor, Umboh, and Williamson~\cite{NaorUW22} obtain a refined guarantee for rooted paths by treating links incident to the known root separately. Here, membership in~$R^*$ depends on the largest requested interval~$M_p$, which is known only at the end. We instead run a given algorithm in phases with geometrically growing budgets and buy representatives at each restart. Once a sufficiently large budget is exceeded, these representatives cover every later request that~$R^*$ covers, leaving only requests covered by~$S^*$ for the final phase. This gives the following conversion.

\begin{samepage}
\begin{theorem}\label{thm:refined-guarantee}
    Let $\cA$ be a deterministic $\rho(n)$-competitive algorithm for online rooted chain covering with $n$ vertices.
    Then there is a deterministic algorithm for online rooted chain covering that, for every instance with root~$p$, outputs a solution
    $F$ such that every feasible solution $F^*$ for the same instance satisfies
    $c(F) \le 48  (c(R^*) + \rho(n) \cdot c(S^*))$,
    where
    $M_p$ is the largest requested interval,
    $R^* = \delta_{F^*}(M_p)$, and $S^* = F^* \setminus R^*$. For an empty request sequence, set $R^*=\emptyset$ and $S^*=F^*$.
\end{theorem}
\end{samepage}

The algorithm does not need to know $\rho(n)$, $F^*$, or~$M_p$.
The proof does not use that link costs are powers of~$2$, so the theorem holds for arbitrary nonnegative costs if $\cA$ is $\rho(n)$-competitive for such costs.
Moreover, the resulting algorithm runs in polynomial time if $\cA$ does, since it starts at most one new copy of~$\cA$ per request and its remaining steps only scan sets of links and requests.

Let $\cA$ be a deterministic $\rho(n)$-competitive algorithm for online rooted chain covering with $n$ vertices, for example \Cref{alg:interval-chain}, and let $\alpha:=\rho(n) \ge 1$.
\Cref{alg:refined-interval-chain} gives the details of our algorithm, where $\cost(\cA)$ denotes the total cost of the links bought by the current copy of~$\cA$.
We call the sequence $\sigma$ of requests sent to one copy $\cA_\sigma$ of $\cA$ a \emph{phase}, including the request that causes its restart, if any, and we call the value $B_\sigma$ of the variable $B$ during this phase its \emph{budget}.
During a phase, the algorithm runs a copy of $\cA$ on the requests that are not yet covered and buys the links that $\cA$ buys, as long as their total cost stays within the budget~$B$ of the phase.
When the links that $\cA$ buys for a request $I$ would raise $\cost(\cA)$ above~$B$, the algorithm discards them, buys a cheapest link $\ell_I$ covering~$I$ instead, which exists because $L$ covers~$I$, and restarts~$\cA$ with a new budget~$B'$.
Every feasible solution pays at least $c_{\ell_I}$ to cover~$I$, so we may set $B'$ to at least $c_{\ell_I}$, which lets the new budget pay for~$\ell_I$, and we also double~$B$ to make the budgets grow geometrically.
At the restart, the algorithm also buys representatives (\Cref{lem:nested-covers}) for the largest interval~$M$ of the phase among the links that cover~$M$ and cost at most~$B'$.
These representatives do not count toward the budget of the new phase.

\suppressfloats[t]
\begingroup
\LinesNotNumbered
\DontPrintSemicolon
\begin{algorithm}[t]
    \caption{Refined algorithm for online rooted chain covering}\label{alg:refined-interval-chain}
    \SetKwBlock{OnRequest}{Upon arrival of interval $I$}{}
    Initialize $F=\{\ell \in L \mid c_\ell = 0\}$, $B=0$, and a fresh copy of $\cA$ with $\cost(\cA)=0$.\;
    \OnRequest{
        \If{$I$ is not covered by $F$}{
            Send $I$ to $\cA$, which buys new links $P$ and serves $I$ in its own copy of the instance.\;
            Update $\cost(\cA)\gets \cost(\cA)+c(P)$.\;
            \eIf{$\cost(\cA)\le B$}{
                Buy the links in $P$, that is, $F\gets F\cup P$, to serve $I$. \;
            }{
                Let $M$ be the largest interval sent to $\cA$ since the last restart, including $I$.\;
                Buy a cheapest link $\ell_I\in\argmin_{\ell\in\delta_L(I)}c_\ell$, that is, $F\gets F\cup\{\ell_I\}$. \;
                Set $B':=\max\{2B,c_{\ell_I}\}$. \;
                Apply \Cref{lem:nested-covers} to $A_{\le B'}(M):=\{\ell\in\delta_L(M)\mid c_\ell\le B'\}$ to obtain $D\subseteq A_{\le B'}(M)$.\;
                Buy all links in $D$, that is, $F\gets F\cup D$.\;
                Start a fresh copy of $\cA$, set $\cost(\cA)\gets0$, and update $B\gets B'$.\;
            }
        }
    }
\end{algorithm}
\endgroup

\Cref{alg:refined-interval-chain} serves every request.
If $\cost(\cA)\le B$, then $F$ contains every link of the current copy of~$\cA$, so $F$ covers~$I$, and otherwise $\ell_I$ covers~$I$.

\paragraph{Competitive analysis.}
A phase ends with a restart when $\cost(\cA)>B$.
Every phase except the last one ends with a restart, and we count the last copy of~$\cA$ as the final phase even if it receives no requests.
For a phase~$\sigma$, we let $\cost(\cA_\sigma)$ denote the total cost of the links that $\cA_\sigma$ buys during~$\sigma$, including the links discarded at its restart.

\begin{lemma}\label{lem:final-phase}
    If a phase $\sigma$ with budget $B_\sigma \ge c(R^*) + \alpha \cdot c(S^*)$ ends with a restart, the next phase is the final phase.
\end{lemma}

\begin{proof}
    Since $F^*$ is a feasible solution, it covers $M_p$, and thus $R^* \neq \emptyset$.
    Let $\ell^* =uv \in R^*$, with $u \in M_p$, be a link covering the smallest requested interval covered by $R^*$.
    As observed above, $\ell^*$ covers every requested interval that $R^*$ covers, and $S^*$ covers all other requested intervals.

    Now consider phase $\sigma$ and its copy $\cA_\sigma$ of~$\cA$.
    If no requested interval in $\sigma$ is covered by $\ell^*$, then $S^*$ covers every request in $\sigma$, that is, $S^*$ is a feasible solution to $\sigma$.
    Since $\cA$ is deterministic, its guarantee applies to $\sigma$, although $\sigma$ depends on earlier decisions of our algorithm.
    Hence,
    \[
    \cost(\cA_\sigma) \le \alpha \cdot c(S^*) \le c(R^*) + \alpha \cdot c(S^*) \le B_\sigma \ ,
    \]
    where the final inequality uses the assumption of the lemma.
    This contradicts the assumption that $\sigma$ ends with a restart, which requires $\cost(\cA_\sigma) > B_\sigma$.
    Thus, there exists an interval $J \in \sigma$ that is covered by $\ell^*$.

    Let $M_\sigma$ denote the largest interval of $\sigma$, which is the interval $M$ chosen at the restart, and let $\sigma'$ be the phase after~$\sigma$, which exists since $\sigma$ ends with a restart.
    Since $J$ is covered by $\ell^*$ and $J \subseteq M_\sigma \subseteq M_p$, we have $u \in J \subseteq M_\sigma$ and $v \notin M_p \supseteq M_\sigma$.
    Hence, $\ell^*$ also covers $M_\sigma$.
    Moreover, as $\ell^* \in R^*$, we have
    \[
    c_{\ell^*} \le c(R^*) \le c(R^*) + \alpha \cdot c(S^*) \le B_\sigma \le B_{\sigma'} \ ,
    \]
    where the last inequality holds because the budget $B_{\sigma'}$ of~$\sigma'$ is the value of~$B'$ computed at the end of~$\sigma$, which is at least $2B_\sigma$.
    Thus, $\ell^* \in A_{\le B_{\sigma'}}(M_\sigma)$.
    Let $D$ be the set of representatives obtained from $A_{\le B_{\sigma'}}(M_\sigma)$ when phase $\sigma$ ends.
    \Cref{fig:final-phase} illustrates how these representatives ensure that no interval that $R^*$ covers reaches the next copy of~$\cA$.

    Every requested interval $K$ that arrives after phase~$\sigma$ contains~$p$, and $K \subseteq M_\sigma$ or $M_\sigma \subseteq K$.
    Hence, if $\ell^*$ covers~$K$, then so does $D$ by \Cref{lem:nested-covers}, since $\ell^* \in A_{\le B_{\sigma'}}(M_\sigma)$.
    As $D\subseteq F$ from the restart on, no interval sent to $\cA_{\sigma'}$ is covered by~$\ell^*$, and hence none is covered by~$R^*$.
    Since $F^*$ is feasible, $S^*$ covers every interval sent to $\cA_{\sigma'}$, so
    \[
    \cost(\cA_{\sigma'}) \le \alpha \cdot c(S^*) \le c(R^*) + \alpha \cdot c(S^*) \le B_\sigma \le B_{\sigma'} \ .
    \]
    Hence, the restart condition never holds in~$\sigma'$, so $\sigma'$ is the final phase.
\end{proof}

\begin{figure}[tbp]
    \centering
    \begin{tikzpicture}[paper/diagram, y=0.85cm]
    \foreach \panel/\shift in {0/0,1/8} {
        \begin{scope}[xshift=\shift cm]
            \foreach \i in {1,...,25} {
                \coordinate (v\i) at ({0.25*(\i-1)},0);
            }
            \ifnum\panel=0
                \node[paper/label] at (3,-3.20) {(a) Coverage by $F^*$};
                \def\comparisonopacity{1}
            \else
                \node[paper/label] at (3,-3.20) {(b) Coverage after buying $D$};
                \def\comparisonopacity{0.4}
            \fi

            \draw[paper/edge] (v1) -- (v25);
            \draw[paper/internal optimal link, draw opacity=\comparisonopacity] (v9)
                .. controls +(0,0.42) and +(0,0.42) ..
                node[paper/annotation, text=col4, text opacity=\comparisonopacity,
                    midway, above=3pt] {$S^*$} (v12);
            \draw[paper/rooted optimal link,
                line width=0.7pt, draw opacity=\comparisonopacity] (v16)
                .. controls +(0,0.99) and +(0,0.99) .. (v22);
            \draw[paper/rooted optimal link, draw opacity=\comparisonopacity] (v13)
                .. controls +(0,1.41) and +(0,1.41) ..
                node[paper/annotation, text opacity=\comparisonopacity,
                    midway, above=3pt] {$\ell^*\in R^*$} (v24);
            \ifnum\panel=1
                \foreach \u/\v/\height/\name in {12/20/0.74/d_2,15/25/0.44/d_4} {
                    \draw[paper/bought link] (v\u)
                        .. controls +(0,\height) and +(0,\height) ..
                        node[paper/annotation, text=paperBought, fill=white, midway, above=3pt] {$\name$} (v\v);
                }
            \fi
            \foreach \i in {1,...,25} {
                \node[paper/vertex] at (v\i) {};
            }
            \node[paper/terminal] at (v9) {};
            \foreach \i/\name in {9/p,13/u,24/v} {
                \node[paper/label, anchor=base, inner sep=0pt]
                    at ({0.25*(\i-1)},-0.42) {$\name$};
            }
            \draw[paper/pivot guide, draw=black!35] (2,-0.55) -- (2,-2.60);
            \foreach \left/\right/\depth in {6/18/-2.10,4/21/-2.60} {
                \foreach \endpoint in {{0.25*(\left-1)-0.10},{0.25*(\right-1)+0.10}} {
                    \draw[paper/boundary guide] ({\endpoint},-0.08) -- ({\endpoint},\depth);
                }
            }

            \foreach \left/\right/\height/\rooted in {
                9/9/-0.78/0,8/10/-1.03/0,7/11/-1.28/0,
                7/14/-1.60/1,6/16/-1.85/1,6/18/-2.10/1,
                5/20/-2.35/1,4/21/-2.60/1%
            } {
                \ifnum\rooted=1
                    \def\intervalcolour{col2}
                \else
                    \def\intervalcolour{paperNeutral}
                \fi
                \def\barwidth{1.1pt}
                \ifnum\right=18 \def\barwidth{1.7pt}\fi
                \ifnum\right=21 \def\barwidth{1.7pt}\fi
                \draw[paper/interval bar, draw=\intervalcolour, line width=\barwidth]
                    ({0.25*(\left-1)-0.10},\height) -- ({0.25*(\right-1)+0.10},\height)
                    ({0.25*(\left-1)-0.10},{\height-0.07}) -- ({0.25*(\left-1)-0.10},{\height+0.07})
                    ({0.25*(\right-1)+0.10},{\height-0.07}) -- ({0.25*(\right-1)+0.10},{\height+0.07});
            }
            \node[paper/label, text=col2, anchor=west] at (4.49,-2.10) {$M_\sigma$};
            \node[paper/label, text=col2, anchor=west] at (5.24,-2.60) {$M_p$};

            \draw[paper/edge, draw=black!35]
                (2.73,-0.70) -- (2.83,-0.70) -- (2.83,-1.36) -- (2.73,-1.36);
            \draw[paper/edge, draw=col2]
                (6.04,-1.52) -- (6.14,-1.52) -- (6.14,-2.68) -- (6.04,-2.68);
            \ifnum\panel=0
                \node[paper/annotation, text=col4, anchor=west, align=left, fill=white] at (2.97,-1.03)
                    {covered\\by $S^*$};
                \node[paper/annotation, anchor=west, align=left] at (6.28,-2.10)
                    {covered\\by $R^*$};
            \else
                \node[paper/annotation, anchor=west, align=left, fill=white] at (2.97,-1.03)
                    {may be sent\\to $\cA_{\sigma'}$};
                \node[paper/annotation, text=paperBought, anchor=west, align=left] at (6.28,-2.10)
                    {covered\\by $D$};
            \fi
        \end{scope}
    }
    \draw[paper/transition] (6.75,0.25) -- (7.40,0.25);
\end{tikzpicture}
    \caption{A restart after phase $\sigma$ with $B_\sigma \ge c(R^*) + \alpha \cdot c(S^*)$ (\Cref{lem:final-phase}), with intervals ordered by containment.
    (a) The black links form $R^*$ and the pink link forms $S^*$.
    The link $\ell^*$ covers every request that $R^*$ covers (orange), and $S^*$ covers the others (gray).
    (b) The blue representatives $D=\{d_2,d_4\}$ for $M_\sigma$, chosen among the links that cover $M_\sigma$ and cost at most the budget $B_{\sigma'}$ of the next phase $\sigma'$, cover all orange requests, so only gray requests can reach $\cA_{\sigma'}$.}
    \label{fig:final-phase}
\end{figure}
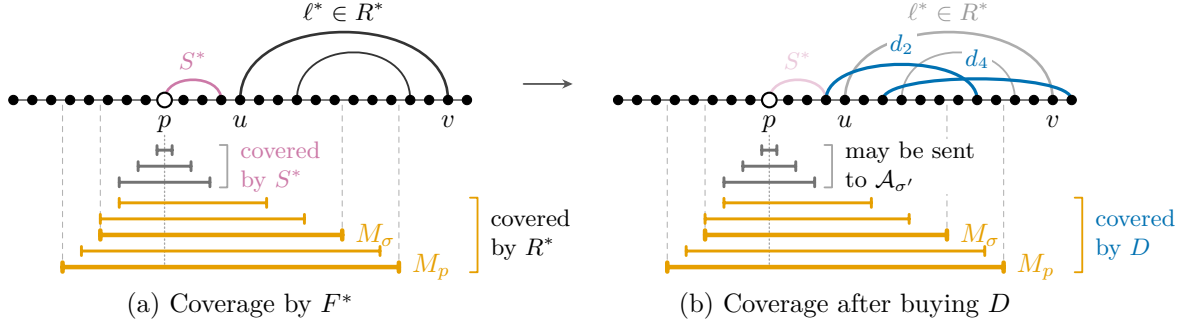

\begin{proof}[Proof of \Cref{thm:refined-guarantee}]
    If no restart occurs, the algorithm incurs no cost, and the claim is immediate.
    Let $B_{\mathrm{fin}}$ be the budget of the final phase.
    If every phase $\sigma$ satisfies $B_\sigma \le c(R^*) + \alpha \cdot c(S^*)$, then $B_{\mathrm{fin}} \le c(R^*) + \alpha \cdot c(S^*)$.

    Otherwise, let $\sigma$ be the first phase with $B_\sigma > c(R^*) + \alpha \cdot c(S^*)$.
    Since the first phase has budget~$0$, the phase $\sigma$ is not the first one.
    Let $\tau$ be the phase before~$\sigma$, so that $B_{\tau} \le c(R^*) + \alpha \cdot c(S^*)$ by the choice of $\sigma$, and let $I'$ denote the final request of~$\tau$, which causes the restart.
    Since $F^*$ is a feasible solution, we have $c_{\ell_{I'}} = \min_{\ell \in \delta_L(I')} c_\ell \le c(F^*) \le c(R^*) + \alpha \cdot c(S^*)$.
    Thus,
    \[
        B_\sigma = \max\{2B_{\tau}, c_{\ell_{I'}}\} \le 2(c(R^*) + \alpha \cdot c(S^*)) \ .
    \]
    By \Cref{lem:final-phase}, there is at most one phase after phase $\sigma$.
    If $\sigma$ is the final phase, then ${B_{\mathrm{fin}} = B_\sigma} \le 2(c(R^*) + \alpha \cdot c(S^*))$.
    Otherwise, let $I$ be the final requested interval of phase $\sigma$. As above, we have $c_{\ell_I} \le c(F^*) \le c(R^*) + \alpha \cdot c(S^*)$, and thus
    \[
        B_{\mathrm{fin}} = \max\{2 B_{\sigma}, c_{\ell_I}\} \le 4(c(R^*) + \alpha \cdot c(S^*)) \ .
    \]
    In summary, the final phase has budget $B_{\mathrm{fin}} \le 4(c(R^*) + \alpha \cdot c(S^*))$ in any case.

    Let $\sigma_1,\ldots,\sigma_h$ be the phases in order of their occurrence.
    The first phase has budget~$0$ and thus incurs no cost before its restart.
    Every later phase has positive budget, since a request $I$ that causes a restart is not covered by the zero-cost links in $F$, and hence $c_{\ell_I}>0$.
    For $i\ge2$, the phase $\sigma_i$ starts with the restart at the end of $\sigma_{i-1}$, where we buy a cheapest covering link and the set $D$ of representatives, which contains at most four links by \Cref{lem:nested-covers}.
    These purchases cost at most $B_{\sigma_i} + 4B_{\sigma_i} = 5B_{\sigma_i}$.
    During $\sigma_i$, the links bought from $\cA$ cost at most $B_{\sigma_i}$ in total.
    Hence, we charge a cost of at most $6B_{\sigma_i}$ to $\sigma_i$, and this accounts for all costs, since a new phase starts after every restart.
    Since $B_{\sigma_{i+1}} \ge 2 B_{\sigma_{i}}$ for all $i \in \{2,\ldots,h-1\}$ and $B_{\sigma_h}=B_{\mathrm{fin}}$, we have
    \[
        c(F) \le 6 \sum_{i=2}^h B_{\sigma_i} \le 6 \cdot 2 \cdot B_{\mathrm{fin}} \le 48 \cdot (c(R^*) + \alpha \cdot c(S^*)) \ . \qedhere
    \]
\end{proof}

\subsection{The Weighted Competitive Ratio}\label{sec:log}

In every subinstance of \Cref{alg:weighted}, we run the conversion of \Cref{alg:refined-interval-chain} with \Cref{alg:interval-chain} as its subroutine~$\cA$. We combine its guarantee with the two counting lemmas from \Cref{sec:log-squared}.

\begin{proof}[Proof of \Cref{thm:main-weighted}]
Fix an instance and an optimal solution~$F^*$, and use the notation of \Cref{sec:log-squared}. By \Cref{thm:interval-chain-augmentation}, \Cref{alg:interval-chain} is $\alpha$-competitive with $\alpha\in O(\log n)$.
For every $p \in V \setminus \{r\}$, the largest interval of $\cC_p$ is $M_p$, so the sets $R^*_p$ and $S^*_p$ are exactly the sets $R^*$ and $S^*$ of \Cref{thm:refined-guarantee} for the feasible solution $F^*_p$ of $\cC_p$.
Therefore, \Cref{thm:refined-guarantee} implies that $c(F_p) \le 48(c(R_p^*) + \alpha \cdot c(S_p^*))$.

By \Cref{lem:active-rooted-count}, every link $\ell \in F^*$ belongs to $R^*_p\subseteq F^*_p$ for at most $2 \lceil \log_2 n \rceil$ vertices $p \in V \setminus \{r\}$.
Thus, $\sum_{p \in V \setminus \{r\}} c(R_p^*) \le 2 \lceil \log_2 n \rceil \cdot c(F^*)$.
Since the sets $S^*_p\subseteq F^*$ are pairwise disjoint by \Cref{lem:inner-disjoint}, we have $\sum_{p \in V \setminus \{r\}} c(S_p^*) \le c(F^*)$.

Altogether, we have
\begin{align*}
c(F)\le\sum_{p\in V\setminus\{r\}}c(F_p)
\le 48\left(\sum_{p\in V\setminus\{r\}}c(R_p^*)+\alpha\sum_{p\in V\setminus\{r\}}c(S_p^*)\right) &\le 48\bigl(2\lceil\log_2 n\rceil+\alpha\bigr)c(F^*) \\*
&=O(\log n)\cdot c(F^*) \ .
\end{align*}
Rounding the link costs incurs only another factor of~$2$, so the same asymptotic bound holds for arbitrary nonnegative costs.\footnote{Explicitly, \Cref{thm:interval-chain-augmentation} gives $\alpha=8\lceil\log_2 n\rceil+18$, and hence a ratio of $480\lceil\log_2 n\rceil+864$ for rounded costs and $960\lceil\log_2 n\rceil+1728$ for arbitrary costs.}
Finally, \Cref{alg:weighted} runs in polynomial time.
Every iteration of its while loop covers the interval that it exposes, so no interval is exposed twice.
Hence, the exposed intervals form a laminar family of distinct nonempty subsets of the $(n-1)$-element set $V\setminus\{r\}$ by \Cref{lem:exposed-laminar}, and there are at most $2(n-1)-1=2n-3$ iterations in total.
Every iteration takes polynomial time, since an exposed interval and its pivot can be found by enumerating all intervals, \Cref{alg:interval-chain} raises $y_I$ by the smallest slack of a dual constraint of a link in $\delta_L(I)$, and the remaining steps of \Cref{alg:refined-interval-chain}, including the computation of the representatives of \Cref{lem:nested-covers}, only scan sets of links.
\end{proof}

Applying \Cref{prop:connectivity-to-cycle} gives a deterministic polynomial-time $O(\log n)$-competitive algorithm for online WCAP, proving \Cref{thm:main-weighted-connectivity}.

\section{Online Cycle Augmentation}\label{sec:unweighted}

In this section, we prove \Cref{thm:main-unweighted,thm:main-unweighted-lower-bound}.
Garg and Subramanian~\cite{GargS26} reduce online CycAP to online Steiner forest in a graph whose Steiner nodes are the links, adjacent when they cross, following offline reductions of CAP to Steiner tree~\cite{BasavarajuFGMRS14,ByrkaGJ23}, and obtain an $O(\log n)$-competitive algorithm by using the online Steiner forest algorithm of Berman and Coulston~\cite{BermanC97}.
We use the same crossing structure, the overlap graph, but replace the Steiner forest algorithm.

\subsection{The Overlap Graph}\label{sec:overlap}

Fix an instance of online CycAP with cycle $C=(V,E)$ on $n$ vertices and link set~$L$.
All links have unit cost, so the cost of a set of links is its cardinality.
We assume that $C = v_1 - v_2 - \cdots - v_n - v_1$ and take indices modulo~$n$.
Two distinct links $\ell,\ell' \in L$ \emph{cross} if they share at least one endpoint or if they have four distinct endpoints that alternate along~$C$.

\begin{definition}[Overlap Graph]
	The \emph{overlap graph} of a CycAP instance with cycle $C$ and links $L$ is the graph $X$ on~$L$ in which two distinct links are adjacent if and only if they cross.
\end{definition}

\begin{figure}[tbp]
    \centering
    \begin{tikzpicture}[
    paper/diagram, scale=0.6,
    paper/link vertex/.style={paper/vertex, minimum size=6pt},
    paper/link name/.style={paper/edge label, text=black!55},
    paper/bought name/.style={paper/edge label, text=paperBought}
]
    \def\cycleradius{2.6}
    \def\cycleangles{90,50,22,-6,-34,-62,-90,-130,-155,180,150,120}
    \def\linklabels{%
        \node[paper/link name] at (1.625,1.340) {$\ell_5$};
        \node[paper/link name] at (-1.525,-1.615) {$\ell_6$};
        \node[paper/link name] at (-1.810,0.785) {$\ell_7$};
        \node[paper/link name] at (0.455,-2.035) {$\ell_8$};
        \node[paper/link name] at (-2.020,0.110) {$\ell_9$};
        \node[paper/bought name] at (1.475,0.500) {$\ell_1$};
        \node[paper/bought name] at (1.760,0.035) {$\ell_2$};
        \node[paper/bought name] at (1.130,-1.345) {$\ell_3$};
    }
    \begin{scope}[xshift=0cm]
        \foreach \a [count=\i] in \cycleangles {
            \coordinate (v\i) at (\a:\cycleradius);
        }
        \draw[paper/edge] (0,0) circle[radius=\cycleradius];
        \node[paper/label, text=black!65] at (0,0.25) {$C$};

        \foreach \u/\v in {2/4,8/10,10/12,5/7,9/11} {
            \draw[paper/link] (v\u) -- (v\v);
        }
        \foreach \u/\v in {1/4,3/6,5/8,7/10} {
            \draw[paper/bought link] (v\u) -- (v\v);
        }
        \linklabels
        \node[paper/bought name] at (-0.820,-1.330) {$\ell_4$};

        \foreach \i in {1,...,12} {
            \node[paper/vertex] at (v\i) {};
        }
        \foreach \i in {1,10} {
            \node[paper/terminal] at (v\i) {};
        }
        \node[paper/label] at (90:\cycleradius+0.36) {$s$};
        \node[paper/label] at (180:\cycleradius+0.50) {$t$};
        \node[paper/label] at (0,-\cycleradius-0.50) {(a) Cycle and links};
    \end{scope}

    \begin{scope}[xshift=6.2cm]
        \coordinate (l1) at (-2.50,0.30);
        \coordinate (l2) at (-0.80,-0.50);
        \coordinate (l3) at (0.90,0.30);
        \coordinate (l4) at (2.60,-0.50);
        \coordinate (l5) at (-2.50,-1.30);
        \coordinate (l6) at (2.60,1.60);
        \coordinate (l7) at (4.70,0.40);
        \coordinate (l8) at (0.90,-1.60);
        \coordinate (l9) at (3.05,0.30);

        \foreach \u/\v in {
            l5/l1,l5/l2,l6/l3,l6/l4,l7/l4,l6/l7,l8/l2,l8/l3,l8/l4,l9/l6,l9/l7,l9/l4%
        } {
            \draw[paper/edge, draw=black!40] (\u) -- (\v);
        }
        \draw[paper/bought link] (l1) -- (l2) -- (l3) -- (l4);

        \foreach \i in {1,...,4} {
            \node[paper/link vertex, draw=paperBought, fill=paperBought] at (l\i) {};
        }
        \foreach \i in {1,3} {
            \node[paper/label, text=paperBought, above=5pt] at (l\i) {$\ell_{\i}$};
        }
        \foreach \i in {2,4} {
            \node[paper/label, text=paperBought, below=5pt] at (l\i) {$\ell_{\i}$};
        }
        \foreach \i in {5,...,9} {
            \node[paper/link vertex, draw=black!55, fill=black!55] at (l\i) {};
        }
        \node[paper/label, text=black!55, below=5pt] at (l5) {$\ell_5$};
        \node[paper/label, text=black!55, above=5pt] at (l6) {$\ell_6$};
        \node[paper/label, text=black!55, right=5pt] at (l7) {$\ell_7$};
        \node[paper/label, text=black!55, below=5pt] at (l8) {$\ell_8$};
        \node[paper/label, text=black!55, above right=3pt] at (l9) {$\ell_9$};
        \node[paper/label] at (1.10,-\cycleradius-0.50) {(b) Overlap graph $X$};
    \end{scope}

    \begin{scope}[xshift=15.6cm]
        \foreach \a [count=\i] in \cycleangles {
            \coordinate (v\i) at (\a:\cycleradius);
        }
        \foreach \from/\to in {50/50,-90/-90,-155/-240} {
            \path[paper/interval] (\from+8:\cycleradius+0.24)
                arc[start angle=\from+8, end angle=\to-8, radius=\cycleradius+0.24]
                -- (\to-8:\cycleradius-0.24)
                arc[start angle=\to-8, end angle=\from+8, radius=\cycleradius-0.24] -- cycle;
        }
        \draw[paper/edge] (0,0) circle[radius=\cycleradius];

        \foreach \u/\v in {2/4,8/10,10/12,5/7,9/11,7/10} {
            \draw[paper/link] (v\u) -- (v\v);
        }
        \foreach \u/\v in {1/4,3/6,5/8} {
            \draw[paper/bought link] (v\u) -- (v\v);
        }
        \linklabels
        \node[paper/link name] at (-0.820,-1.330) {$\ell_4$};

        \foreach \i in {1,...,12} {
            \node[paper/vertex] at (v\i) {};
        }
        \foreach \i in {3,4,5,6,8} {
            \node[paper/vertex, draw=paperBought, fill=paperBought, minimum size=4.5pt] at (v\i) {};
        }
        \foreach \i in {1,10} {
            \node[paper/terminal] at (v\i) {};
        }
        \node[paper/label] at (90:\cycleradius+0.36) {$s$};
        \node[paper/label] at (180:\cycleradius+0.50) {$t$};
        \node[paper/label] at (-130:\cycleradius+0.36) {$x$};
        \node[paper/label, text=paperInterval] at (128:\cycleradius+0.55) {$J$};
        \node[paper/label] at (0,-\cycleradius-0.50) {(c) Gaps of $A=\{\ell_1,\ell_2,\ell_3\}$};
    \end{scope}
\end{tikzpicture}
    \caption{Links on a cycle, their overlap graph, and gaps.
    (a) The set $F$ of blue links serves $\{s,t\}$.
    (b) The graph $X[F]$ has a path between links at $s$ and~$t$ (\Cref{lem:overlap-characterization}).
    (c) The gaps of $A=\{\ell_1,\ell_2,\ell_3\}$ (shaded), where $x$ and~$s$ are adjacent to the gap~$J$.
    The links $\ell_4$ and~$\ell_6$ cover~$J$, whereas $\ell_7$ and~$\ell_9$ lie in~$J$.}
    \label{fig:overlap-anchors}
\end{figure}
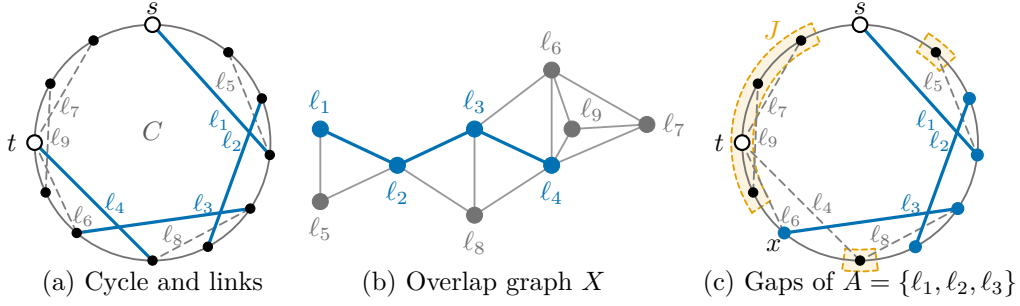

The overlap graph is also used for offline cycle augmentation by Traub and Zenklusen, who call it the link intersection graph~\cite{TraubZ23}.
For a set of links $F \subseteq L$, let $X[F]$ denote the subgraph of $X$ induced by~$F$, and let $V(F)\subseteq V$ denote the set of endpoints of the links in~$F$.
We call the link sets of the connected components of $X[F]$ the \emph{components} of~$F$.
Since all links incident to a common vertex pairwise cross, every cycle vertex $w\in V(F)$ is incident only to links of a single component of~$F$, which we call the \emph{component of $F$ at~$w$}.

\begin{lemma}[{\cite[Proposition~C.1]{GargS26}}]\label{lem:overlap-characterization}
	A set of links $F\subseteq L$ serves a request $\{s,t\}$ if and only if $X[F]$ contains a path (possibly of length zero) from a link incident to $s$ to a link incident to~$t$.
\end{lemma}

In other words, $F$ serves $\{s,t\}$ if and only if $s,t\in V(F)$ and the components of $F$ at $s$ and at $t$ coincide, as \Cref{fig:overlap-anchors}(a) and~(b) illustrate.
For completeness, we prove \Cref{lem:overlap-characterization} in \Cref{app:overlap}.

For a nonempty set $A\subseteq L$, we call the vertex sets of the connected components of $C-V(A)$ the \emph{gaps} of~$A$, as shown in \Cref{fig:overlap-anchors}(c).
Every gap $J$ of $A$ is a circular interval $\{v_i,v_{i+1},\ldots,v_j\}$ of $C$, and the two vertices $v_{i-1}$ and $v_{j+1}$ outside~$J$ are \emph{adjacent} to~$J$.
They belong to $V(A)$, as $J$ is the vertex set of a connected component of~$C-V(A)$, and they are distinct, since otherwise $J$ would contain all vertices but one, whereas $|V(A)|\ge2$ as $A$ is nonempty.
If $X[A]$ is connected, $s\in V(A)$, and $t$ lies in a gap~$J$ of~$A$, then $J$ is a circular interval that separates $s$ and~$t$ and that no link of~$A$ covers, so every set of links that serves $\{s,t\}$ contains a link that covers~$J$ by \Cref{lem:cut-covering}.
The next lemma shows that every such link crosses a link of~$A$, such as $\ell_4$ or~$\ell_6$ in \Cref{fig:overlap-anchors}(c).

\begin{lemma}[Gap Lemma]\label{lem:gap}
    Let $A\subseteq L$ be a nonempty set of links such that $X[A]$ is connected.
    Then every link that covers a gap of $A$ crosses a link of~$A$.
\end{lemma}

\begin{proof}
    Let $uv\in L$ be a link that covers a gap $J$ of~$A$, where $u\in J$ and $v\notin J$, and suppose for a contradiction that $uv$ crosses no link of~$A$.
    As $u\in J$, we have $u\notin V(A)$, so $uv\notin A$, and as $uv$ crosses no link of~$A$, it shares no endpoint with a link of~$A$, so $v\notin V(A)$.
    Removing $u$ and $v$ splits $C$ into two paths $Q_1$ and $Q_2$, one of which may be empty.
    Every link of $A$ has both endpoints on $Q_1$ or both endpoints on $Q_2$, as it does not cross~$uv$.
    Moreover, a link with both endpoints on $Q_1$ and a link with both endpoints on $Q_2$ do not cross, since their endpoints are distinct and do not alternate along~$C$.
    As $X[A]$ is connected, we may thus assume that all links of $A$ have both endpoints on~$Q_1$.
    Then the $u$--$v$ path in~$C$ through~$Q_2$ has no vertex in~$V(A)$, so $u$ and $v$ lie in the same gap of~$A$, namely in~$J$, which contradicts $v\notin J$.
\end{proof}

\subsection{The Algorithm}\label{sec:unweighted-algorithm}

Our algorithm maintains the set $F$ of bought links, which is initially empty.
Upon a request $\{s,t\}$, it first considers $w=s$ and then $w=t$.
If $w\notin V(F)$, it \emph{activates}~$w$: it buys a link incident to~$w$ that crosses a link of~$F$ if such a link exists, and an arbitrary link incident to~$w$ otherwise, which then forms a new component of~$F$.
A link incident to~$w$ exists by \Cref{lem:overlap-characterization}, since $L$ is feasible.

Afterwards, as long as the components of $F$ at $s$ and at $t$ differ, the algorithm grows the component $K$ of $F$ at~$s$.
If possible, it merges $K$ with another component of $F$ by buying at most two links.\footnote{The algorithm remains correct if merges may consist of up to $h\ge2$ links, but our analysis then only gives $(h+2)$-competitiveness.}
Otherwise, it expands $K$ toward~$t$.
It walks along the cycle from $t$ in both directions until it first meets a vertex incident to a link of~$N[K]$, where $N[A]$ denotes the closed neighborhood of a set~$A$ of links in~$X$, that is, the set of links that belong to~$A$ or cross a link of~$A$.
At each of the two vertices where it stops, it buys a link of $N[K]\setminus K$ incident to this vertex, unless a link of~$K$ is already incident to it.

Formally, merges and expansions are defined as follows.

\begin{definition}[Merge and Expansion]\label{def:merge-expansion}
    Let $F$ be the current solution, let $\{s,t\}$ be the current request with $s,t\in V(F)$, and let $K$ be the component of $F$ at~$s$.
    Suppose that $t\notin V(K)$.
    \begin{itemize}
        \item A \emph{merge} for $K$ is a set $P\subseteq L\setminus F$ with $|P|\le2$ such that $X[K\cup P\cup D]$ is connected for some component $D\neq K$ of~$F$.
        \item Suppose that no merge for $K$ exists.
        Let $x$ and $y$ be the two vertices adjacent to the gap $J$ of $N[K]$ with~$t\in J$.
        For each $z\in\{x,y\}$, let $\ell_z\in N[K]$ be a link incident to~$z$, chosen from~$K$ if possible.
        Then $Y:=\{\ell_x,\ell_y\}\setminus K$ is the \emph{expansion} of $K$ toward~$t$.
    \end{itemize}
\end{definition}

The expansion is well defined.
Indeed, since $K$ is a component of~$F$, every bought link that crosses a link of~$K$ belongs to~$K$, that is, $F\cap N[K]=K$.
Suppose for a contradiction that a link $\ell\in N[K]$ is incident to~$t$.
Then $\ell\notin K$, as $t\notin V(K)$, so $\ell\notin F$, and $\ell$ crosses a link of~$K$.
As $\ell$ also shares the endpoint~$t$ with a link of the component $D\neq K$ of $F$ at~$t$, the set $\{\ell\}$ is a merge for~$K$, contradicting the assumption that no merge for $K$ exists.
Hence, $t$ lies in a gap $J$ of~$N[K]$, and the two vertices $x$ and~$y$ adjacent to~$J$ belong to $V(N[K])$, so the links $\ell_x$ and~$\ell_y$ exist.
Since $F\cap N[K]=K$, the expansion~$Y$ consists of the links among $\ell_x$ and~$\ell_y$ that are not bought yet, and every link of~$Y$ crosses a link of~$K$, as $Y\subseteq N[K]\setminus K$.
\Cref{fig:expansion} illustrates an expansion, and \Cref{alg:unweighted} summarizes the algorithm.

\begin{figure}[tbp]
    \centering
    \begin{tikzpicture}[
    paper/diagram, scale=0.75,
    paper/feasible link/.style={paper/internal optimal link, line width=1.1pt}
]
    \foreach \stage/\shift in {0/0,1/6.6,2/13.2} {
        \begin{scope}[xshift=\shift cm]
            \foreach \i in {0,...,17} {
                \coordinate (v\i) at ({90-20*\i}:2.0);
            }
            \ifnum\stage<2
                \path[paper/interval] (102:2.24)
                    arc[start angle=102, end angle=178, radius=2.24]
                    -- (178:1.76)
                    arc[start angle=178, end angle=102, radius=1.76] -- cycle;
            \fi
            \draw[paper/edge] (0,0) circle[radius=2.0];

            \ifnum\stage=0
                \draw[paper/link] (v0) --
                    node[paper/edge label, text=black!55, pos=0.42, below left=0.5pt] {$\ell_y$} (v3);
                \draw[paper/link] (v8) --
                    node[paper/edge label, text=black!55, pos=0.5, above right=2pt] {$\ell_x$} (v13);
            \else
                \draw[paper/bought link] (v0) --
                    node[paper/edge label, text=paperBought, pos=0.42, below left=0.5pt] {$\ell_y$} (v3);
                \draw[paper/bought link] (v8) --
                    node[paper/edge label, text=paperBought, pos=0.5, above right=2pt] {$\ell_x$} (v13);
            \fi
            \ifnum\stage<2
                \draw[paper/feasible link] (v11) --
                    node[paper/edge label, text=col4, pos=0.17, right=3.5pt] {$f$} (v14);
                \draw[paper/feasible link] (v14) -- (v16);
            \else
                \draw[paper/bought link] (v11) --
                    node[paper/edge label, text=paperBought, pos=0.17, right=3.5pt] {$f$} (v14);
                \draw[paper/bought link] (v14) -- (v16);
            \fi

            \draw[paper/bought link] (v1) -- (v7);
            \draw[paper/bought link] (v4) -- (v10);
            \draw[paper/bought link] (v15) -- (v17);
            \ifnum\stage<2
                \node[paper/label, text=paperBought] at (0.58,-0.18) {$K$};
                \node[paper/label, text=paperBought] at (128:1.50) {$D$};
            \fi

            \foreach \i in {0,...,17} {
                \node[paper/vertex] at (v\i) {};
            }
            \foreach \i in {1,15} {
                \node[paper/terminal] at (v\i) {};
            }
            \node[paper/label] at (70:2.36) {$s$};
            \node[paper/label] at (150:2.50) {$t$};
            \ifnum\stage<2
                \node[paper/label] at (190:2.36) {$x$};
                \node[paper/label] at (90:2.36) {$y$};
                \node[paper/label, text=paperInterval] at (122:2.55) {$J$};
            \fi
        \end{scope}
    }
    \draw[paper/transition] (2.75,0) -- (3.85,0);
    \draw[paper/transition] (9.35,0) -- (10.45,0);
    \node[paper/label] at (0,-2.50) {(a) Before the expansion};
    \node[paper/label] at (6.6,-2.50) {(b) After the expansion};
    \node[paper/label] at (13.2,-2.50) {(c) After the merge};
\end{tikzpicture}
    \caption{An expansion of $K$ toward $t$ and the subsequent merge (\Cref{lem:expansion}).
    Bought links are blue, dashed links cross links of $K$, pink links belong to a feasible solution, and other links are omitted.
    (a) The gap $J$ of $N[K]$ that contains $t$.
    (b) The link $f$ covers $J$ and crosses a bought link only after the expansion.
    (c) Buying the two pink links of~(b), which form a merge, connects $s$ and~$t$.}
    \label{fig:expansion}
\end{figure}
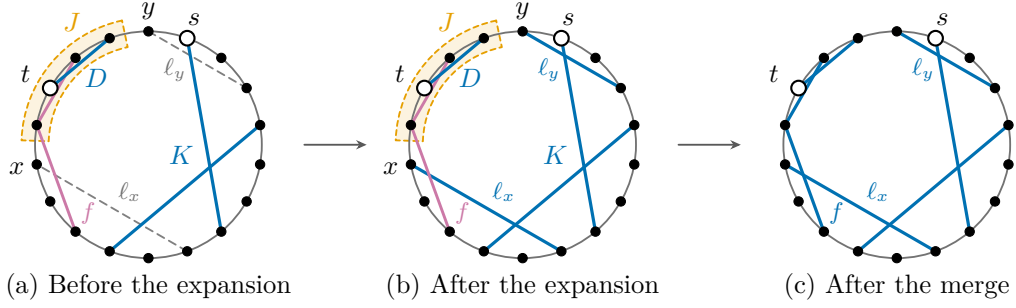

\begingroup
\LinesNotNumbered
\DontPrintSemicolon
\begin{algorithm}[t]
\caption{Algorithm for online CycAP}\label{alg:unweighted}
    \SetKwBlock{OnRequest}{Upon arrival of request $\{s,t\}$}{}
    Initialize $F=\emptyset$.\;
    \OnRequest{
        \For{$w=s,t$}{
            \If{$w\notin V(F)$}{
                Buy a link incident to $w$, preferring a link that crosses a link of $F$.\;
            }
        }
        \While{the components of $F$ at $s$ and $t$ differ}{
            Let $K$ be the component of $F$ at $s$.\;
            \eIf{there is a merge $P$ for $K$}{
                Buy $P$.\;
            }{
                Buy the expansion of $K$ toward $t$.\;
            }
        }
    }
\end{algorithm}
\endgroup

To pay for an expansion~$Y$, we use a link of an optimal solution in $N[F\cup Y]\setminus N[F]$, that is, a link that belongs to or crosses a bought link after the expansion but not before.
Since $N[F]$ only grows, every optimal link is used at most once in this way.
The next lemma shows that such a link exists.

\begin{lemma}[Expansion Lemma]\label{lem:expansion}
    In the setting of \Cref{def:merge-expansion}, suppose that no merge for $K$ exists, and let $Y$ be the expansion of $K$ toward~$t$.
    Then every set $F^*\subseteq L$ that serves $\{s,t\}$ contains a link of $N[F\cup Y]\setminus N[F]$.
\end{lemma}

\begin{proof}
    Let $J$, $x$, $y$, $\ell_x$, and $\ell_y$ be as in \Cref{def:merge-expansion}, and let $F^*$ serve $\{s,t\}$.
    The gap $J$ of~$N[K]$ is a circular interval, that is, a minimum cut of~$C$, and it separates $s$ and~$t$, as $t\in J$ and $s\in V(K)\subseteq V(N[K])$.
    Hence, some link $f\in F^*$ covers~$J$ by \Cref{lem:cut-covering}.
    We show that $f\in N[F\cup Y]\setminus N[F]$.

    The graph $X[K\cup Y]$ is connected, since $X[K]$ is connected and every link of~$Y$ crosses a link of~$K$.
    Moreover, $J$ is also a gap of~$K\cup Y$, since no link of $K\cup Y\subseteq N[K]$ has an endpoint in~$J$, whereas the links $\ell_x,\ell_y\in K\cup Y$ are incident to the vertices $x$ and~$y$ adjacent to~$J$.
    Hence, \Cref{lem:gap} with $A=K\cup Y$ shows that $f$ crosses a link of $K\cup Y$.
    As no link of $N[K]$ has an endpoint in~$J$, we have $f\notin N[K]$, so $f$ crosses a link $\ell\in Y$, and thus $f\in N[F\cup Y]$.

    Suppose for a contradiction that $f\in N[F]$, and recall that $\ell\notin F$ crosses a link of~$K$.
    We distinguish two cases, which \Cref{fig:merge-cases} illustrates.

    \textbf{Case $f\in F$.}
    Let $D$ be the component of $F$ that contains~$f$.
    Then $D\neq K$, since $f\notin N[K]$, and as $\ell$ crosses a link of $K$ and the link $f$ of~$D$, the graph $X[K\cup\{\ell\}\cup D]$ is connected.
    Hence, $\{\ell\}$ is a merge for~$K$.

    \textbf{Case $f\notin F$.}
    Since $f\in N[F]$, the link $f$ crosses a link of~$F$, and we let $D$ be the component of $F$ that contains such a link.
    Then $D\neq K$, since $f\notin N[K]$, and as $\ell$ crosses a link of $K$ and the link~$f$, which crosses a link of~$D$, the graph $X[K\cup\{\ell,f\}\cup D]$ is connected.
    Since neither $\ell$ nor $f$ belongs to~$F$, the set $\{\ell,f\}$ is a merge for~$K$.
    This is the only place where we use that merges may consist of two links.

    In both cases, a merge for $K$ exists, which contradicts our assumption.
\end{proof}

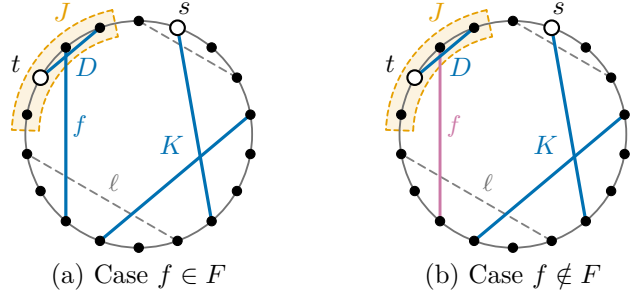
\begin{figure}[tbp]
    \centering
    \begin{tikzpicture}[
    paper/diagram, scale=0.75,
    paper/feasible link/.style={paper/internal optimal link, line width=1.1pt}
]
    \foreach \stage/\shift in {0/0,1/6.6} {
        \begin{scope}[xshift=\shift cm]
            \foreach \i in {0,...,17} {
                \coordinate (v\i) at ({90-20*\i}:2.0);
            }
            \path[paper/interval] (102:2.24)
                arc[start angle=102, end angle=178, radius=2.24]
                -- (178:1.76)
                arc[start angle=178, end angle=102, radius=1.76] -- cycle;
            \draw[paper/edge] (0,0) circle[radius=2.0];

            \draw[paper/link] (v0) -- (v3);
            \draw[paper/link] (v8) --
                node[paper/edge label, text=black!55, pos=0.5, above right=2pt] {$\ell$} (v13);
            \ifnum\stage=0
                \draw[paper/bought link] (v11) --
                    node[paper/edge label, text=paperBought, pos=0.55, right=1.5pt] {$f$} (v16);
            \else
                \draw[paper/feasible link] (v11) --
                    node[paper/edge label, text=col4, pos=0.55, right=1.5pt] {$f$} (v16);
            \fi
            \draw[paper/bought link] (v1) -- (v7);
            \draw[paper/bought link] (v4) -- (v10);
            \draw[paper/bought link] (v15) -- (v17);
            \node[paper/label, text=paperBought] at (0.58,-0.18) {$K$};
            \node[paper/label, text=paperBought] at (128:1.50) {$D$};

            \foreach \i in {0,...,17} {
                \node[paper/vertex] at (v\i) {};
            }
            \foreach \i in {1,15} {
                \node[paper/terminal] at (v\i) {};
            }
            \node[paper/label] at (70:2.36) {$s$};
            \node[paper/label] at (150:2.50) {$t$};
            \node[paper/label, text=paperInterval] at (122:2.55) {$J$};
        \end{scope}
    }
    \node[paper/label] at (0,-2.50) {(a) Case $f\in F$};
    \node[paper/label] at (6.6,-2.50) {(b) Case $f\notin F$};
\end{tikzpicture}
    \caption{The two cases in the proof of \Cref{lem:expansion} if $f\in N[F]$, drawn as in \Cref{fig:expansion}.
    (a) If $f\in F$, then $\{\ell\}$ is a merge.
    (b) Otherwise, $f$ crosses a link of~$D$, and $\{\ell,f\}$ is a merge.}
    \label{fig:merge-cases}
\end{figure}

\paragraph{Correctness.}
After the activations for a request $\{s,t\}$, we have $s,t\in V(F)$, and as long as the components of $F$ at $s$ and at $t$ differ, we have $t\notin V(K)$, so \Cref{def:merge-expansion} applies.
Every merge is nonempty, since $X[K\cup D]$ is disconnected for distinct components $K$ and $D$ of~$F$, and every expansion $Y$ is nonempty, since \Cref{lem:expansion} applied to $F^*=L$ yields a link of $N[F\cup Y]\setminus N[F]$.
Hence, every iteration of the while loop buys a new link, so the loop ends after at most $|L|$ iterations in total.
Afterwards, the components of $F$ at $s$ and at $t$ coincide, so $F$ serves $\{s,t\}$ by \Cref{lem:overlap-characterization}, and as $F$ only grows, it also serves all previous requests.
Finally, the algorithm runs in polynomial time, since the overlap graph, the components of~$F$, the $O(|L|^2)$ candidates for a merge, and the gaps of $N[K]$ can be computed in polynomial time.

\paragraph{Competitive analysis.}
We now show that \Cref{alg:unweighted} is $4$-competitive, which proves \Cref{thm:main-unweighted}.

\begin{proof}[Proof of \Cref{thm:main-unweighted}]
    As argued above, \Cref{alg:unweighted} is a deterministic polynomial-time algorithm for online CycAP.
    Fix a request sequence and an optimal solution~$F^*$.
    We show that, at the end, the set $F$ of bought links satisfies $|F|\le4|F^*|$.
 
    Let $\kappa(F)$ denote the number of components of~$F$.
    Let $m_{\mathrm{act}}$ be the number of activations, and let $m_{\mathrm{new}}$ be the number of activations that increase~$\kappa(F)$.
    Moreover, let $m_{\mathrm{merge}}$ and $m_{\mathrm{exp}}$ be the numbers of merges and expansions bought by the algorithm, respectively.
    Since every activation buys one link and every merge and every expansion buys at most two links, we have
    \[
        |F|\le m_{\mathrm{act}}+2m_{\mathrm{merge}}+2m_{\mathrm{exp}} \ .
    \]

    Every requested vertex is activated at most once, since it belongs to $V(F)$ afterwards, and it belongs to $V(F^*)$ by \Cref{lem:overlap-characterization}.
    Hence, $m_{\mathrm{act}}\le|V(F^*)|\le2|F^*|$.

    Next, we show that $m_{\mathrm{merge}}\le m_{\mathrm{new}}$.
    An activation increases $\kappa(F)$ by one if the bought link crosses no link of~$F$, and it does not increase $\kappa(F)$ otherwise.
    Every link of an expansion crosses a link of~$K$, so expansions do not increase $\kappa(F)$ either.
    Every merge decreases $\kappa(F)$ by at least one, since $K$, $P$, and $D$ then lie in a common component.
    Since $\kappa(\emptyset)=0$, we obtain $0\le\kappa(F)\le m_{\mathrm{new}}-m_{\mathrm{merge}}$ at the end, so $m_{\mathrm{merge}}\le m_{\mathrm{new}}$.

    Finally, we show that $m_{\mathrm{new}}+m_{\mathrm{exp}}\le|F^*|$.
    As $F$ only grows, so does $F^*\cap N[F]$, which has at most $|F^*|$ links.
    Hence, it suffices to show that every expansion and every activation that increases~$\kappa(F)$ adds a link to~$F^*\cap N[F]$.

    \begin{itemize}
	\item
    \textbf{Case of an expansion.}
    Let $Y$ be an expansion of the component $K$ of $F$ at~$s$ toward~$t$ for a request $\{s,t\}$.
    Then no merge for $K$ exists, and $F^*$ serves $\{s,t\}$, so $F^*$ contains a link of $N[F\cup Y]\setminus N[F]$ by \Cref{lem:expansion}.
    Buying $Y$ thus adds this link to $F^*\cap N[F]$.
    \item
    \textbf{Case of an activation that increases~$\kappa(F)$.}
    Let $w$ be the activated vertex, and let $f\in F^*$ be a link incident to~$w$, which exists since $w$ is a requested vertex.
    Before the activation, $f\notin F$, since $w\notin V(F)$.
    Moreover, the bought link crosses no link of~$F$, as the activation increases~$\kappa(F)$.
    Since the algorithm prefers links that cross a link of~$F$, no link incident to $w$ crosses a link of~$F$, and in particular, $f$ does not.
    Hence, $f\notin N[F]$ before the activation, but $f\in N[F]$ afterwards, since $f$ is the bought link or shares the endpoint $w$ with it.
    \end{itemize}

    Altogether, we obtain
    \[
        |F|\le m_{\mathrm{act}}+2m_{\mathrm{merge}}+2m_{\mathrm{exp}}
        \le 2|F^*|+2m_{\mathrm{new}}+2m_{\mathrm{exp}}
        \le 2|F^*|+2|F^*|
        = 4|F^*| \ .
    \]
    This shows that \Cref{alg:unweighted} is $4$-competitive and completes the proof of the theorem.
\end{proof}

\subsection{The Lower Bound}\label{sec:unweighted-lower-bound}

We now show that the competitive ratio~$4$ is best possible for deterministic algorithms, which proves \Cref{thm:main-unweighted-lower-bound}.

The links of our instance form bundles that are arranged in a cycle.
The adversary forces the algorithm to buy a link in every bundle on one half of this cycle and three links in almost every bundle on the other half, whereas a single link in every bundle of this half suffices.

\paragraph{The instance.}
Let $m\ge2$ be a power of~$2$.
The cycle $C$ consists of $4m$ groups of consecutive vertices, which appear along $C$ in the order $U_0,W_0,U_1,W_1,\ldots,U_{2m-1},W_{2m-1}$, where $|U_i|=2$ and $|W_i|=3$ for every~$i$.
We take indices modulo~$2m$.
For every~$i$, the \emph{bundle} $L_i$ consists of the six links between $U_i$ and~$W_{i+1}$, which jump over the groups $W_i$ and~$U_{i+1}$.
The set of available links $L$ is the union of all bundles.
Thus, $C$ has $10m$ vertices, $L$ has $12m$ links, and every vertex is incident only to links of a single bundle.
\Cref{fig:bundle-instance} illustrates the instance.

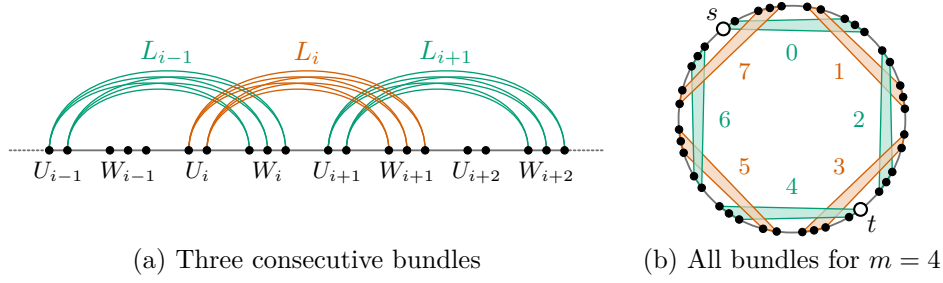
\begin{figure}[tbp]
    \centering
    \begin{tikzpicture}[
    paper/diagram, scale=0.75,
    paper/small vertex/.style={paper/vertex, minimum size=2.8pt},
    paper/arc/.style={line width=0.5pt, out=90, in=90, looseness=1.15},
    paper/main bundle/.style={paper/arc, draw=col5},
    paper/side bundle/.style={paper/arc, draw=col3},
    paper/band/.style={line width=0.5pt, fill opacity=0.7}
]
    \begin{scope}[yshift=-0.55cm]
        \foreach \name/\x in {a1/0, a2/0.32, b1/1.07, b2/1.39, b3/1.71, c1/2.46, c2/2.78,
            d1/3.53, d2/3.85, d3/4.17, e1/4.92, e2/5.24, f1/5.99, f2/6.31, f3/6.63,
            g1/7.38, g2/7.70, h1/8.45, h2/8.77, h3/9.09} {
            \coordinate (\name) at (\x,0);
        }
        \draw[paper/edge, densely dotted] (-0.7,0) -- (a1);
        \draw[paper/edge] (a1) -- (h3);
        \draw[paper/edge, densely dotted] (h3) -- (9.79,0);
        \foreach \p in {a1,a2} { \foreach \q in {d1,d2,d3} {
            \draw[paper/side bundle] (\p) to (\q); } }
        \foreach \p in {e1,e2} { \foreach \q in {h1,h2,h3} {
            \draw[paper/side bundle] (\p) to (\q); } }
        \foreach \p in {c1,c2} { \foreach \q in {f1,f2,f3} {
            \draw[paper/main bundle] (\p) to (\q); } }
        \foreach \name in {a1,a2,b1,b2,b3,c1,c2,d1,d2,d3,e1,e2,f1,f2,f3,g1,g2,h1,h2,h3} {
            \node[paper/small vertex] at (\name) {};
        }
        \foreach \lab/\x in {U_{i-1}/0.16, W_{i-1}/1.39, U_i/2.62, W_i/3.85,
            U_{i+1}/5.08, W_{i+1}/6.31, U_{i+2}/7.54, W_{i+2}/8.77} {
            \node[paper/annotation, anchor=north] at (\x,-0.16) {$\lab$};
        }
        \node[paper/label, text=col3] at (2.08,1.72) {$L_{i-1}$};
        \node[paper/label, text=col5] at (4.55,1.72) {$L_i$};
        \node[paper/label, text=col3] at (7.0,1.72) {$L_{i+1}$};
        \node[paper/label] at (4.545,-1.95) {(a) Three consecutive bundles};
    \end{scope}

    \begin{scope}[xshift=13.1cm]
        \draw[paper/edge] (0,0) circle[radius=2.0];
        \foreach \i in {0,...,7} {
            \pgfmathsetmacro{\tu}{123.75-45*\i}
            \pgfmathsetmacro{\tw}{123.75-45*\i-67.5}
            \ifodd\i \def\bandcolour{col5}\else \def\bandcolour{col3}\fi
            \path[paper/band, draw=\bandcolour, fill=\bandcolour!30] ({\tu+3.5}:2.0) -- ({\tu-3.5}:2.0) -- ({\tw+6}:2.0)
                -- ({\tw-6}:2.0) -- cycle;
        }
        \foreach \i in {0,...,7} {
            \pgfmathsetmacro{\tu}{123.75-45*\i}
            \pgfmathsetmacro{\tw}{123.75-45*\i-22.5}
            \foreach \d in {3.5,-3.5} { \node[paper/small vertex] at ({\tu+\d}:2.0) {}; }
            \foreach \d in {6,0,-6} { \node[paper/small vertex] at ({\tw+\d}:2.0) {}; }
            \ifodd\i \def\bandcolour{col5}\else \def\bandcolour{col3}\fi
            \node[paper/annotation, text=\bandcolour] at ({90-45*\i}:1.18) {$\i$};
        }
        \node[paper/terminal, minimum size=4.8pt] at (127.25:2.0) {};
        \node[paper/terminal, minimum size=4.8pt] at (-52.75:2.0) {};
        \node[paper/label] at (127.25:2.33) {$s$};
        \node[paper/label] at (-52.75:2.33) {$t$};
        \node[paper/label] at (0,-2.5) {(b) All bundles for $m=4$};
    \end{scope}
\end{tikzpicture}
    \caption{The instance of \Cref{thm:main-unweighted-lower-bound}.
    (a) Three consecutive bundles, drawn with the cycle cut open.
    (b) The bundles for $m=4$, each drawn as a band and labeled with its index, and the vertices $s$ and $t$ of the first request.
    }
    \label{fig:bundle-instance}
\end{figure}

Removing the endpoints of a link $\ell\in L_i$ splits $C$ into two paths, one containing $W_i\cup U_{i+1}$ and the other containing all vertices outside $U_i\cup W_i\cup U_{i+1}\cup W_{i+1}$, including $U_{i-1}$ and~$W_{i+2}$, as $m\ge2$.
Every link of $L_{i-1}$ joins $U_{i-1}$ with~$W_i$, and every link of $L_{i+1}$ joins $U_{i+1}$ with~$W_{i+2}$, so these links have an endpoint on each of the two paths and cross~$\ell$.
In contrast, the links of all other bundles have both endpoints on the second path, so they do not cross~$\ell$.
Hence, two links of distinct bundles cross if and only if their bundles are consecutive in the \emph{bundle cycle} $L_0 - L_1 - \cdots - L_{2m-1} - L_0$.
For distinct bundles $L_a$ and~$L_b$, the two paths between them in the bundle cycle are $L_a - L_{a+1} - \cdots - L_b$ and $L_b - L_{b+1} - \cdots - L_a$.
The following lemma translates \Cref{lem:overlap-characterization} to the bundle cycle.

\begin{lemma}[Bundle Path Lemma]\label{lem:bundle-paths}
    Let $x$ and $y$ be vertices that are incident to links of distinct bundles $L_a$ and~$L_b$, respectively.
    Then a set of links $F\subseteq L$ serves $\{x,y\}$ if and only if $F$ contains a link incident to~$x$, a link incident to~$y$, and a link of every bundle on one of the two paths between $L_a$ and $L_b$ in the bundle cycle.
\end{lemma}

\begin{proof}
    Suppose first that $F$ contains such links.
    We choose one such link in every bundle on the path, taking the link incident to~$x$ in~$L_a$ and the link incident to~$y$ in~$L_b$.
    Since links of consecutive bundles cross, the chosen links form a path in $X[F]$ between these two links, so $F$ serves $\{x,y\}$ by \Cref{lem:overlap-characterization}.

    Conversely, suppose that $F$ serves $\{x,y\}$.
    By \Cref{lem:overlap-characterization}, $X[F]$ contains a path from a link incident to~$x$, which belongs to~$L_a$, to a link incident to~$y$, which belongs to~$L_b$.
    Any two consecutive links of this path cross, so they belong to the same bundle or to consecutive bundles.
    Hence, the bundles that contain links of this path induce a connected subgraph of the bundle cycle that contains $L_a$ and~$L_b$, and thus one of the two paths between them.
\end{proof}

\paragraph{The adversary's request sequence.}
Fix a deterministic online algorithm, and let $F$ denote the set of links bought by the algorithm so far.
All requests contain a common vertex $s\in U_0$, which is incident only to links of~$L_0$.
We call a bundle \emph{empty} if it contains no link of~$F$.
For $i\neq0$, a \emph{visit} of the bundle $L_i$ consists of two requests.
First, the adversary requests $\{s,u\}$ for a vertex $u\in U_i$, and then it requests $\{s,w\}$ for a vertex $w\in W_{i+1}$.
Each time, it chooses a vertex that is not incident to a link of~$F$ if possible, and an arbitrary one otherwise.
An offline solution can buy the single link $uw\in L_i$, which is incident to both requested vertices.
\Cref{fig:bundle-visit} illustrates a visit.

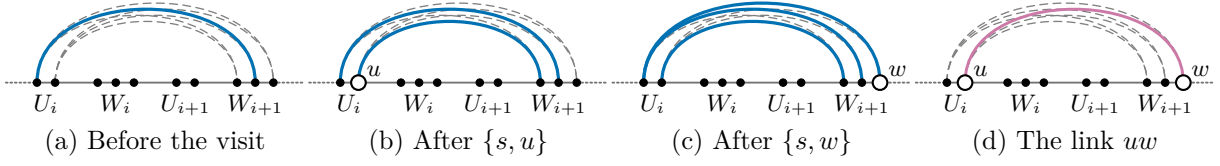
\begin{figure}[tbp]
    \centering
    \begin{tikzpicture}[
    paper/diagram, scale=0.75,
    paper/small vertex/.style={paper/vertex, minimum size=2.8pt},
    paper/arc/.style={out=90, in=90, looseness=1.15},
    paper/free arc/.style={paper/arc, paper/link, line width=0.5pt},
    paper/bought arc/.style={paper/arc, paper/bought link},
    paper/feasible arc/.style={paper/arc, paper/internal optimal link, line width=1.1pt}
]
    \foreach \stage/\shift in {0/0,1/5.35,2/10.7,3/16.05} {
        \begin{scope}[xshift=\shift cm]
            \foreach \name/\x in {c1/0, c2/0.32, d1/1.07, d2/1.39, d3/1.71, e1/2.46,
                e2/2.78, f1/3.53, f2/3.85, f3/4.17} {
                \coordinate (\name) at (\x,0);
            }
            \draw[paper/edge, densely dotted] (-0.55,0) -- (c1);
            \draw[paper/edge] (c1) -- (f3);
            \draw[paper/edge, densely dotted] (f3) -- (4.72,0);
            \foreach \p in {c1,c2} { \foreach \q in {f1,f2,f3} {
                \draw[paper/free arc] (\p) to (\q); } }
            \ifnum\stage<3
                \draw[paper/bought arc] (c1) to (f2);
            \fi
            \ifnum\stage>0 \ifnum\stage<3
                \draw[paper/bought arc] (c2) to (f1);
            \fi\fi
            \ifnum\stage=2
                \draw[paper/bought arc] (c1) to (f3);
            \fi
            \ifnum\stage=3
                \draw[paper/feasible arc] (c2) to (f3);
            \fi
            \foreach \name in {c1,c2,d1,d2,d3,e1,e2,f1,f2,f3} {
                \node[paper/small vertex] at (\name) {};
            }
            \ifnum\stage=1
                \node[paper/terminal] at (c2) {};
                \node[paper/edge label, anchor=south west] at (0.42,0.07) {$u$};
            \fi
            \ifnum\stage=2
                \node[paper/terminal] at (f3) {};
                \node[paper/edge label, anchor=south west] at (4.27,0.07) {$w$};
            \fi
            \ifnum\stage=3
                \node[paper/terminal] at (c2) {};
                \node[paper/terminal] at (f3) {};
                \node[paper/edge label, anchor=south west] at (0.42,0.07) {$u$};
                \node[paper/edge label, anchor=south west] at (4.27,0.07) {$w$};
            \fi
            \foreach \lab/\x in {U_i/0.16, W_i/1.39, U_{i+1}/2.62, W_{i+1}/3.85} {
                \node[paper/annotation, anchor=north] at (\x,-0.16) {$\lab$};
            }
        \end{scope}
    }
    \node[paper/label] at (2.085,-1.05) {(a) Before the visit};
    \node[paper/label] at (7.435,-1.05) {(b) After $\{s,u\}$};
    \node[paper/label] at (12.785,-1.05) {(c) After $\{s,w\}$};
    \node[paper/label] at (18.135,-1.05) {(d) The link $uw$};
\end{tikzpicture}
    \caption{A visit of a bundle $L_i$ that is not empty before the visit (\Cref{lem:bundle-visit}).
    Bought links are blue, the other links of $L_i$ are gray and dashed, and links of other bundles are omitted.
    (a) One link of $L_i$ is bought.
    (b) No bought link is incident to~$u$, so serving $\{s,u\}$ requires a new link at~$u$.
    (c) No bought link is incident to~$w$, so serving $\{s,w\}$ requires a new link at~$w$, and $L_i$ ends with three bought links.
    (d) The single pink link $uw$ is incident to both $u$ and~$w$.}
    \label{fig:bundle-visit}
\end{figure}

\begin{lemma}[Visit Lemma]\label{lem:bundle-visit}
    After a visit of a bundle~$L_i$, the set $F$ contains at least three links of~$L_i$ if $L_i$ was not empty before the visit, and at least two links of~$L_i$ otherwise.
\end{lemma}

\begin{proof}
    Every link of $L_i$ is incident to exactly one vertex of~$U_i$ and one vertex of~$W_{i+1}$, and all links incident to these vertices belong to~$L_i$.
    Consider a request $\{s,v\}$ of the visit, let $G\in\{U_i,W_{i+1}\}$ be the group that contains~$v$, and let $b$ be the number of links of~$L_i$ in~$F$ before the request.
    If $b<|G|$, then some vertex of~$G$ is not incident to a link of~$F$, so by the choice of~$v$, neither is~$v$.
    To serve $\{s,v\}$, the algorithm thus buys a link incident to~$v$ by \Cref{lem:bundle-paths}, and this link belongs to~$L_i$.
    Hence, $F$ contains at least $\min\{b+1,|G|\}$ links of~$L_i$ after the request.
    For $G=U_i$ with $|U_i|=2$, this gives at least two links after the request $\{s,u\}$ if $L_i$ was not empty before the visit, and at least one otherwise.
    For $G=W_{i+1}$ with $|W_{i+1}|=3$, it then gives at least three and at least two links, respectively, after the request $\{s,w\}$.
\end{proof}

The adversary proceeds in two stages, which \Cref{fig:bundle-run} illustrates for $m=4$.
\begin{enumerate}
    \item The adversary requests $\{s,t\}$ for a vertex $t\in U_m$.
    By \Cref{lem:bundle-paths}, the algorithm buys a link in every bundle on one of the two paths between $L_0$ and~$L_m$ in the bundle cycle. Choose one such path~$P$, let~$Q$ be the other path, and let~$B$ be the set of indices of the inner bundles of~$Q$, so $B=\{1,\ldots,m-1\}$ or $B=\{m+1,\ldots,2m-1\}$.
    \item The adversary visits the bundles $L_i$ with $i\in B$ in $\log_2 m$ rounds, which perform a binary search over~$B$.
    In round $d\in\{1,\ldots,\log_2 m\}$, it visits the bundles $L_i$ with $i\in B$ for which $i$ is an odd multiple of $m/2^d$, in increasing order of~$i$.
    Every index in $B$ is an odd multiple of $m/2^d$ for exactly one~$d$, so every bundle $L_i$ with $i\in B$ is visited exactly once.
    For example, if $m=16$ and $B=\{1,\ldots,15\}$, the four rounds visit $L_8$, then $L_4$ and~$L_{12}$, then $L_2$, $L_6$, $L_{10}$, and~$L_{14}$, and finally all $L_i$ with odd~$i$.
\end{enumerate}
Since serving the first request of the visit of a bundle~$L_i$ requires a bought link in every bundle on one of the two paths between $L_0$ and~$L_i$ (\Cref{lem:bundle-paths}), the binary-search order ensures that almost all visited bundles are not empty before their visit.

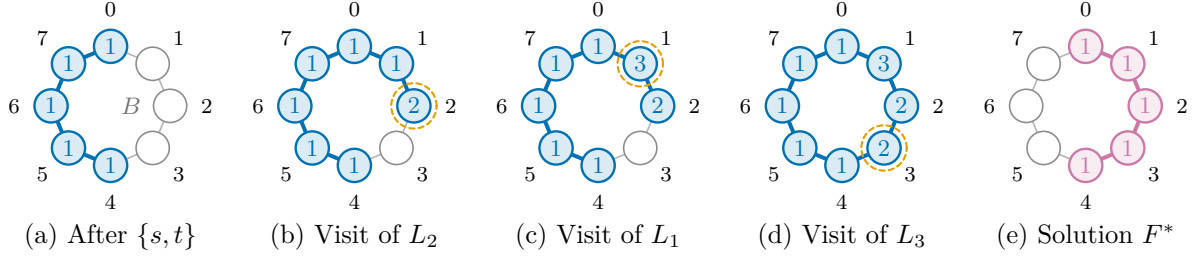
\begin{figure}[tbp]
    \centering
    \begin{tikzpicture}[
    paper/diagram, scale=0.75,
    paper/bundle/.style={
        circle, draw=black!45, fill=white, line width=0.6pt,
        inner sep=0pt, minimum size=12.5pt, font=\footnotesize
    },
    paper/nonempty bundle/.style={
        paper/bundle, draw=paperBought, fill=paperBought!14, text=paperBought,
        line width=0.9pt
    },
    paper/solution bundle/.style={
        paper/bundle, draw=col4, fill=col4!14, text=col4, line width=0.9pt
    },
    paper/bundle edge/.style={draw=black!30, line width=0.6pt},
    paper/bought edge/.style={draw=paperBought, line width=1.5pt},
    paper/solution edge/.style={draw=col4, line width=1.5pt},
    paper/visit mark/.style={draw=paperInterval, line width=0.8pt, dash pattern=on 2pt off 1.5pt},
    paper/index/.style={paper/annotation, text=black, font=\scriptsize}
]
    \foreach \stage/\shift/\counts/\visit/\caption in {
        0/0/{1,0,0,0,1,1,1,1}/-1/{(a) After $\{s,t\}$},
        1/4.3/{1,1,2,0,1,1,1,1}/2/{(b) Visit of $L_2$},
        2/8.6/{1,3,2,0,1,1,1,1}/1/{(c) Visit of $L_1$},
        3/12.9/{1,3,2,2,1,1,1,1}/3/{(d) Visit of $L_3$},
        4/17.2/{1,1,1,1,1,0,0,0}/-1/{(e) Solution $F^*$}%
    } {
        \begin{scope}[xshift=\shift cm]
            \foreach \i [count=\k from 0] in \counts {
                \coordinate (B\k) at ({90-45*\k}:1.05);
                \expandafter\xdef\csname cnt\k\endcsname{\i}
            }
            \foreach \i in {0,...,7} {
                \pgfmathtruncatemacro{\j}{mod(\i+1,8)}
                \draw[paper/bundle edge] (B\i) -- (B\j);
                \pgfmathtruncatemacro{\a}{\csname cnt\i\endcsname}
                \pgfmathtruncatemacro{\b}{\csname cnt\j\endcsname}
                \ifnum\a>0 \ifnum\b>0
                    \ifnum\stage=4
                        \draw[paper/solution edge] (B\i) -- (B\j);
                    \else
                        \draw[paper/bought edge] (B\i) -- (B\j);
                    \fi
                \fi\fi
            }
            \foreach \i in {0,...,7} {
                \pgfmathtruncatemacro{\a}{\csname cnt\i\endcsname}
                \ifnum\a=0
                    \node[paper/bundle] at (B\i) {};
                \else
                    \ifnum\stage=4
                        \node[paper/solution bundle] at (B\i) {$\a$};
                    \else
                        \node[paper/nonempty bundle] at (B\i) {$\a$};
                    \fi
                \fi
                \node[paper/index] at ({90-45*\i}:1.7) {$\i$};
            }
            \ifnum\visit>-1
                \draw[paper/visit mark] (B\visit) circle[radius=0.4];
            \fi
            \ifnum\stage=0
                \node[paper/annotation, text=black!55] at (0.36,0) {$B$};
            \fi
            \node[paper/label] at (0,-2.3) {\caption};
        \end{scope}
    }
\end{tikzpicture}
    \caption{A run of the adversary for $m=4$ on the bundle cycle.
    In (a)--(d), numbers count the bought links of each bundle, and the visited bundle is circled.
    (a) The algorithm serves $\{s,t\}$ with links in $L_4,\ldots,L_7$ and~$L_0$, so $B=\{1,2,3\}$.
    (b) Round~1 visits the empty bundle~$L_2$, and the algorithm also buys a link in~$L_1$.
    (c) Round~2 visits $L_1$, which was not empty and ends with three links.
    (d) Round~2 visits the empty bundle~$L_3$, which ends with two links.
    (e) $F^*$ has one link per bundle on the pink path, so $|F^*|=5$, whereas the algorithm buys $12$ links.
    }
    \label{fig:bundle-run}
\end{figure}

\begin{lemma}\label{lem:binary-search}
    At most $\log_2 m$ bundles are empty before their visit.
\end{lemma}

\begin{proof}
    As there are $\log_2 m$ rounds, it suffices to show that every round visits at most one bundle that is empty before its visit.
    The first round visits a single bundle.
    Consider two bundles $L_i$ and $L_j$ with $i<j$ that are visited in the same round $d\ge2$.
    As $i$ and $j$ are odd multiples of $m/2^d$, the index $q=i+m/2^d$ satisfies $i<q<j$, so $q\in B$, and $q$ is a multiple of $m/2^{d-1}$, so $L_q$ is visited in an earlier round.
    After the algorithm serves the first request $\{s,u\}$ of the visit of~$L_q$, where $u\in U_q$, the set $F$ contains a link of every bundle on one of the two paths between $L_0$ and~$L_q$ in the bundle cycle by \Cref{lem:bundle-paths}.
    As $0<i<q<j<2m$, one of these paths contains~$L_i$, and the other one contains~$L_j$.
    Since this visit of~$L_q$ ends before round~$d$ begins and bought links are never removed, $L_i$ or $L_j$ is not empty when round~$d$ begins, and thus not empty before its visit.
\end{proof}

We can finally complete the proof of \Cref{thm:main-unweighted-lower-bound}.

\begin{proof}[Proof of \Cref{thm:main-unweighted-lower-bound}]
    In the first stage, the algorithm buys at least one link in each of the $m+1$ bundles on~$P$.
    By \Cref{lem:bundle-visit,lem:binary-search}, every bundle $L_i$ with $i\in B$, which the adversary visits in the second stage, ends up with at least three bought links, except for at most $\log_2 m$ of them, which end up with at least two bought links.
    As the inner bundles of~$Q$ do not lie on~$P$ and distinct bundles are disjoint, the number of links that the algorithm buys is at least
    \[
        (m+1)+3(m-1)-\log_2 m = 4m-2-\log_2 m \ .
    \]

    Let $F^*$ consist of a link of $L_0$ incident to~$s$, a link of $L_m$ incident to~$t$, and, for every $i\in B$, the link $uw\in L_i$ between the two vertices $u$ and $w$ requested in the visit of~$L_i$.
    Then $F^*$ consists of $m+1$ links, one in every bundle on~$Q$.
    Every request is of the form $\{s,v\}$, where $v$ is incident to a link of $F^*$ in some bundle $L_j\neq L_0$ on~$Q$.
    Thus, $F^*$ contains a link incident to~$s$, a link incident to~$v$, and a link of every bundle on the subpath of $Q$ between $L_0$ and~$L_j$, so $F^*$ serves $\{s,v\}$ by \Cref{lem:bundle-paths}.
    Finally, the ratio
    \[
        \frac{4m-2-\log_2 m}{m+1}=4-\frac{6+\log_2 m}{m+1}
    \]
    tends to~$4$ as $m$ grows, so no deterministic online algorithm is $\rho$-competitive for a constant $\rho<4$.
\end{proof}

\section{Conclusion}

Our results determine the competitive ratio of online WCAP up to a constant factor and show that the best constant competitive ratio of deterministic algorithms for online CycAP is~$4$, even when all requests share a common vertex.
On cycles, unit costs thus admit a constant competitive ratio, whereas costs in~$\{0,1\}$ already force a competitive ratio of~$\Omega(\log n)$.

Our weighted algorithm buys links only to cover a laminar family of minimum cuts, which it constructs online.
Its main ingredient converts every deterministic competitive algorithm for online rooted chain covering into one that loses only a constant factor on the links of an optimal solution that cover the largest requested interval, although this interval is known only at the end.

For online survivable network design, a gap remains between the randomized $\widetilde{O}(r_{\max}\log^3 n)$ bound for general requirements~\cite{GuptaKR12} and the $\Omega(\log n)$ lower bound.
General requirements bring nonuniform demands and minimum cuts that change as links are bought.
It would be interesting to see whether laminar families constructed online and the conversion for rooted chain covering reduce the logarithmic factors or the dependence on~$r_{\max}$.

For online CycAP, the optimal competitive ratio is open for randomized algorithms and in the random-order model.
Our adversary chooses each requested vertex depending on the links bought so far, so our lower bound applies neither to randomized algorithms nor to the random-order model.
More generally, it is open for which graph classes online CAP admits a constant competitive ratio.

\appendix

\section{Proof of the Reduction to Weighted Cycle Augmentation}\label{app:reduction}

\begin{proof}[Proof of \Cref{prop:connectivity-to-cycle}]
    We compose the reduction of Garg and Subramanian~\cite{GargS26} to WCacAP with the unfolding steps of G\'alvez, Grandoni, Jabal Ameli, and Sornat~\cite{GalvezGAS21}, check that each map preserves link costs and, up to the added zero-cost links, the sets of links that serve each request, and run the given algorithm on the resulting cycle instance.

    If $G$ is $(k+1)$-edge-connected, the empty set serves every request.
    Otherwise, the construction in the proof of~\cite[Proposition~2.3]{GargS26}, applied to a cactus representation with $m=O(n)$ edges, maps the instance to an instance of online WCacAP such that link costs are preserved and a set of links serves a request if and only if its image serves the image of the request~\cite[Claim~A.2]{GargS26}.
    A request whose two vertices have the same image is served by the empty set, since no minimum cut of~$G$ separates them, so we discard it.
    Moreover, we delete every link whose image is a loop, since a loop covers no set, and among parallel links, which cover the same sets, we keep a cheapest one.
    This preserves the above equivalence and the cost of an optimal solution, and the cactus instance satisfies our conventions that requests consist of two distinct vertices, that $L$ is feasible, and that no two links are parallel.
    If $m=2$, then all requests ask for the same pair of vertices, and buying a cheapest link joining them upon the first request is optimal.
    Otherwise, the unfolding steps yield a cycle with $m\ge3$ vertices, since they do not change the number of edges, and they create neither loops nor parallel links.

    Consider one unfolding step at a vertex~$x$ that lies on at least two cycles, which turns an old cactus into a new one.
    It joins two of these cycles into one cycle on which $x$ appears as two copies $x_1$ and~$x_2$, attaches all other edges and all links at~$x$ to~$x_1$, and adds a zero-cost link~$x_1x_2$.
    Contracting $x_1$ and~$x_2$ in the new cactus reverses this step.
    Hence, the minimum cuts of the new cactus that do not separate $x_1$ and $x_2$ are exactly the sets obtained from the minimum cuts of the old cactus by replacing $x$ with $x_1$ and $x_2$ whenever $x$ belongs to the cut, and each is covered by the same links as the old cut.
    Every other minimum cut separates $x_1$ and $x_2$ and is covered by the zero-cost link.
    By \Cref{lem:cut-covering}, a set of links serves a request in the old cactus if and only if, together with $x_1x_2$, it serves, in the new cactus, the request obtained by replacing $x$ with~$x_1$.

    Since all these maps are fixed before the requests arrive, we can run a $\rho(m)$-competitive algorithm for online WCycAP on the mapped requests and buy a preimage of the same cost of every link that it buys, where the added links $x_1x_2$ have no preimage.
    By the above equivalences, this serves every request, and optimal solutions of all instances have the same cost.
    As $m\le\lceil an\rceil$ for a suitable constant $a\ge1$ and $\rho$ is nondecreasing, the resulting algorithm is $\rho(\lceil an\rceil)$-competitive.
    All maps can be computed deterministically in polynomial time, since the cactus representation can be constructed deterministically in polynomial time~\cite{Fleischer98} and the remaining maps are explicit~\cite[Theorem~A.1]{GargS26}, \cite[Theorem~4]{GalvezGAS21}.
    Hence, the resulting algorithm is deterministic if the given one is, and it runs in polynomial time if the given one does.
\end{proof}

\section{Proof of the Overlap Graph Characterization}\label{app:overlap}

We use the notation of \Cref{sec:overlap}.

\begin{proof}[Proof of \Cref{lem:overlap-characterization}]
    Suppose first that $X[F]$ contains such a path.
    Let $S$ be a circular interval that separates $s$ and~$t$, say with $s\in S$ and $t\notin S$, and suppose for a contradiction that $F$ does not cover~$S$.
    Then every link of the path has both endpoints in~$S$ or both endpoints in~$V\setminus S$.
    The first link of the path has both endpoints in~$S$, as it is incident to~$s$, and the last one has both endpoints in~$V\setminus S$, as it is incident to~$t$.
    Hence, the path contains two consecutive links, one with both endpoints in~$S$ and one with both endpoints in~$V\setminus S$.
    These two links cross, whereas their endpoints are distinct and do not alternate along~$C$, since $S$ and $V\setminus S$ are circular intervals, a contradiction.
    Thus, $F$ covers every circular interval that separates $s$ and~$t$, so $F$ serves $\{s,t\}$ by \Cref{lem:cut-covering}.

    Conversely, suppose that $F$ serves $\{s,t\}$.
    Since $F$ covers the circular interval $\{s\}$, some link of $F$ is incident to~$s$, and we let $K$ be the component of $F$ at~$s$.
    Suppose for a contradiction that no link of $K$ is incident to~$t$.
    Then $t$ lies in a gap $J$ of~$K$, and $J$ separates $s$ and~$t$, as $s\in V(K)$.
    No link of $K$ has an endpoint in~$J$, and since $K$ is a component of~$F$, no link of $F\setminus K$ crosses a link of~$K$, so by \Cref{lem:gap}, no link of $F\setminus K$ covers~$J$.
    Hence, $F$ does not cover the circular interval~$J$, which contradicts \Cref{lem:cut-covering}.
    Thus, $K$ contains a link incident to~$s$ and a link incident to~$t$, and since $X[K]$ is connected, $X[F]$ contains a path between them.
\end{proof}

\section*{Independent Work}

During discussions at a workshop~\cite{umboh26}, we learned that another group, comprising Hathcock, Subramanian, Umboh, and Zlatin, was also working on the unweighted cycle augmentation problem. We informed them that we had obtained an \(O(1)\)-competitive algorithm for the cycle augmentation problem. Later, we learned that they had also obtained an \(O(1)\)-competitive algorithm~\cite{umboh26}. We are not aware of their exact results or techniques, and our work was conducted entirely independently of theirs.

\section*{Acknowledgments}

We thank Kevin Schewior for fruitful discussions.

\section*{AI Methodology Statement}

We used Claude Opus 5.5 to improve the writing and the structure of the proofs, to check them, for literature research, and to create all TikZ figures. 

\printbibliography

\end{document}